\pdfoutput=1

\newif\ifsingleColumn
\singleColumnfalse	

\newif\ifArxiv
\Arxivtrue 

\ifsingleColumn
\documentclass[9pt,journal,draftcls,onecolumn]{IEEEtran}
\else
\documentclass[9pt,journal]{IEEEtran}
\fi
\ifCLASSINFOpdf
\else
\fi
\usepackage[english]{babel}
\usepackage{cite}

\usepackage{amsmath}
\usepackage{amsthm}
\usepackage{amsfonts}
\usepackage{amssymb}  
\usepackage{empheq} 
\usepackage{units}
\usepackage{accents}
\newcommand\munderbar[1]{%
  \underaccent{\bar}{#1}}

\usepackage{graphicx}
\usepackage{epstopdf}	

\usepackage{tikz}
\usetikzlibrary{shapes,arrows}
\usepackage{pgfplots}	
\usepackage{pdfpages}

\newtheorem{theorem}{Theorem}[section]
\newtheorem{lemma}[theorem]{Lemma}

\newtheorem{corollary}[theorem]{Corollary}

\newenvironment{remark}[1][Remark]{\begin{trivlist}
\item[\hskip \labelsep {\bfseries #1}]}{\end{trivlist}}

\renewcommand{\qed}{\nobreak \ifvmode \relax \else
      \ifdim\lastskip<1.5em \hskip-\lastskip
      \hskip1.5em plus0em minus0.5em \fi \nobreak
      \vrule height0.75em width0.5em depth0.25em\fi}

\tikzset{every picture/.style={font issue=\footnotesize},
         font issue/.style={execute at begin picture={#1\selectfont}}
        }

\newcommand{\TT}{\ensuremath{\mathsf{\tiny{T}}}}
\newcommand{\T}{^{\TT}}

\newcommand{\diag}{\mathrm{diag}}
\newcommand{\tr}[1]{\mathrm{tr}(#1)}

\newcommand{\comp}[1]{{#1}^{\mathsf{c}}}

\newcommand{\HChange}[1]{#1}

\newcommand{\HChangeN}[1]{#1}

\newcommand{\VecT}[2][]{\left( \begin{array}{c} #1 \\ #2 \end{array} \right)}

\ifArxiv
\usepackage{fancyhdr}
\newcommand{\mytitle}{\vspace*{-5pt}\textbf{Accepted final version.} This is an extended version of an article to appear in the IEEE Transactions on Automatic Control (additional parts in the Appendix).\\ \copyright 2017 IEEE. Personal use of this material is permitted. Permission from IEEE must be obtained for all other uses, in any current or future media, including reprinting/ republishing this material for advertising or promotional purposes, creating new collective works, for resale or redistribution to servers or lists, or reuse of any copyrighted component of this work in other works.}
\fi

\begin{document}
%
\title{Distributed Event-Based State Estimation for \\Networked Systems: An LMI-Approach}
%
%
%

\author{Michael~Muehlebach,~\IEEEmembership{Student-Member,~IEEE,}
        Sebastian~Trimpe,~\IEEEmembership{Member,~IEEE}
\thanks{M. Muehlebach is with the Institute for Dynamic Systems and Control, ETH Zurich, 8092 Zurich, Switzerland, e-mail: michaemu@ethz.ch.}
\thanks{S. Trimpe is with the Autonomous Motion Department at the Max Planck Institute for Intelligent Systems, 72076 T\"ubingen, Germany, e-mail: strimpe@tuebingen.mpg.de.}
\thanks{This work was supported in part by the Swiss National Science Foundation, the Max Planck Society, and the Max Planck ETH Center for Learning Systems.}%
}


%
%

\markboth{IEEE Transactions On Automatic Control,~Vol.~X, No.~Y, September~Z}%
{Muehlebach \MakeLowercase{\textit{et al.}}: Distributed Event-Based State Estimation of Networked Systems: An LMI-Approach}
%



\maketitle

\ifArxiv
\thispagestyle{fancy}	
\pagestyle{empty}
\fi

\begin{abstract}
In this work, a dynamic system is controlled by multiple sensor-actuator agents, each of them commanding and observing parts of the system's input and output. The different agents sporadically exchange data with each other via a common bus network according to local event-triggering protocols. 
From these data, each agent estimates the complete dynamic state of the system and uses its estimate for feedback control.
We propose a synthesis procedure for designing the agents' state estimators and the event triggering thresholds.
The resulting distributed and event-based control system is guaranteed to be stable and to satisfy a predefined estimation performance criterion. 
\HChange{The approach is applied to the control of a vehicle platoon, where the method's trade-off between performance and communication, and the scalability in the number of agents is demonstrated.}
\end{abstract}


%
\IEEEpeerreviewmaketitle

\vspace*{-5pt}\section{Introduction}\label{Sec:Intro}
The majority of today's control systems are implemented on digital hardware with a periodic exchange of data between the various system's components, e.g.\ reading sensor values, providing actuation commands, etc. While periodic information exchange simplifies the analysis of the resulting control systems, it is fundamentally limited: system resources such as computation and communication are used at predetermined time instants irrespective of the current state of the system, or the information content of the data to be passed between the components. This is not the case with event-based strategies, where information is exchanged or processed only when certain events indicate that an update would be favorable, for instance, to improve the control or estimation performance. System resources are therefore only used when necessary. As a consequence, event-based communication for control, estimation, and optimization is an active and growing area of research, see e.g.\HChange{\cite{Le11,HeJoTa12,GrHiJuEtAl14,Ca14,MiskowiczEvent,EventbasedBroadcast}} 
and references therein.

In this work, we consider event-based communication for a distributed control system, where multiple sensor and actuator agents observe and control a dynamic system and exchange data via a common bus network, as shown in Fig.~\ref{Fig:NCS}. In previous work \cite{Tr12,Sebastian-CDC}, an architecture for distributed state estimation with event-based communication between the agents was proposed. Each agent consists of three main components: the controller computes actuation commands based on the information obtained from the state estimator; the event generator (EG) decides whether local measurements are transmitted over the common bus network and shared with all agents; and the state estimator reconstructs the system's state based on the measurements communicated over the bus network.
The event generator compares the current measurement to the prediction of the measurement by the state estimator for making effective transmit decisions. 
\HChange{The architecture is distributed due to the fact that transmit decisions, state estimates, and control inputs are computed locally.}
The common bus is a key element of the proposed architecture as it facilitates information sharing between all components.
Bus systems as assumed herein are common in industry automation \cite{Th05}, and have recently also been proposed for multi-hop low-power wireless networks \cite{FeZiMoTh12}.

\begin{figure}[tb]
\centering
\graphicspath{{media/}}
\includegraphics[scale=.8]{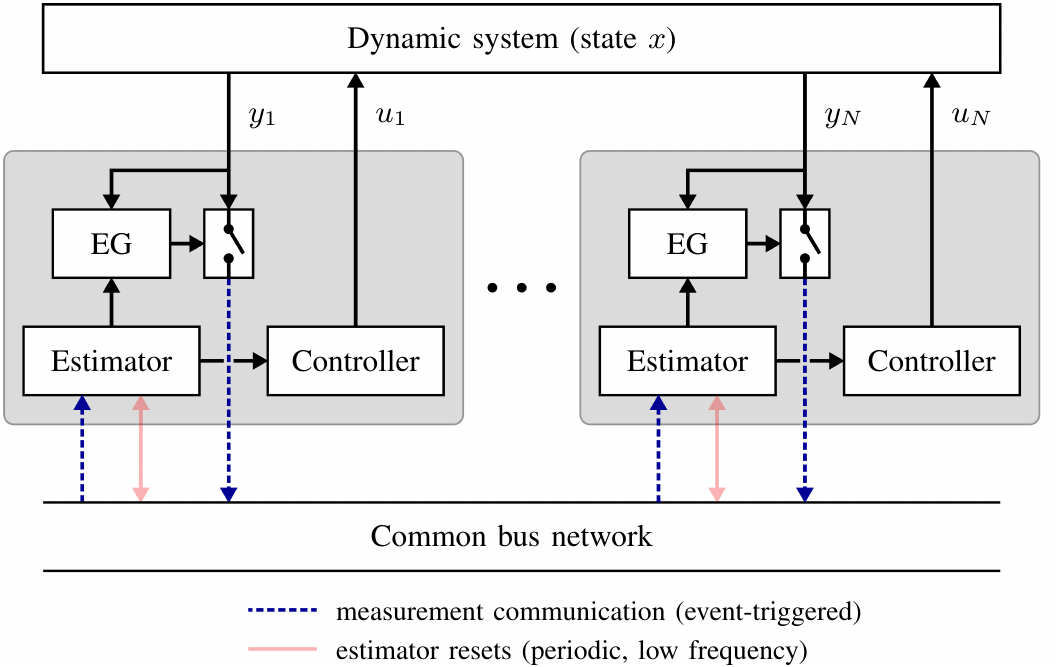}
\caption{Networked control system 
considered in this paper.
Each agent observes part of the system state $x$ through local sensors $y_i$ and sends commands $u_i$ to its local actuator. 
Event-triggered communication is indicated by dashed arrows, while periodic communication is shown by solid ones. 
The periodic estimator resets can be avoided under certain conditions (to be made precise later).  The common-bus architecture is motivated by commonly used field-bus systems \cite{Th05}, such as CAN on the Balancing Cube \cite{trimpe2012balancing}, as well as recent wireless systems \cite{FeZiMoTh12}. 
}
\label{Fig:NCS}
\end{figure}

The approach in \cite{Tr12,Sebastian-CDC} has been shown to be effective for reducing measurement communication in experiments on the Balancing Cube test bed \cite{trimpe2012balancing}, which has a network architecture as in Fig.~\ref{Fig:NCS}.
The method in \cite{Tr12,Sebastian-CDC} relies on a distributed and event-based implementation that emulates a given centralized observer and controller design. In \cite{Tr12}, closed-loop stability is shown in an ideal scenario with perfect communication (no delay or packet loss) and identically initialized state estimates.
To guarantee closed-loop stability also for the case where state estimates may differ (e.g.\ due to packet losses), additional periodic estimator resets are introduced in \cite{Sebastian-CDC}. Both approaches require periodic communication of the inputs.

\HChange{In this work, a modified design is proposed, which further reduces network load by avoiding the communication of the control inputs altogether and under favorable circumstances (to be made precise) also the periodic estimator resets.}
In contrast to \cite{Tr12,Sebastian-CDC}, which obtain the estimator gains from a centralized Luenberger observer design and the event triggering thresholds by manual tuning, we synthesize observer gains \emph{and} triggering thresholds specifically for the distributed and event-based estimation problem. A flexible performance objective is derived, such that the state estimator design can be formulated as an optimization problem. 
The optimization is augmented with linear matrix inequalities (LMIs) imposing closed-loop stability. As a result, both, the state estimator and the event generator, are designed by solving convex optimization problems, \cite{BoydLMI}. 



Preliminary results of those herein were presented in the conference papers \cite{muehlebach2015lmi} and \cite{muehlebach2015guaranteed}, which focused on stability and performance, respectively. \HChange{The main extensions of this article include a less conservative stability condition for the inter-agent error; a relaxation of the LMI-design that scales linearly instead of exponentially in the number of agents; 
new simulation examples; and the unified presentation of previous stability and performance results.}


\emph{Related Work:}
%
%
Distributed event-based state estimation designs based on LMI formulations are also proposed in \cite{zhang2015event,Zhang2015,YaZhZhYa14,MiTiFiDiJoRu12}, whose relation to this work is discussed next.  
For a general overview and references on event-based state estimation, the reader is referred to the reviews in \cite{MiskowiczEvent,ShShCh16,SebastianArxiv}.
 
While herein filtering performance is considered in terms of an $\mathcal{H}_2$ index (e.g. like in the steady-state Kalman filter), \cite{zhang2015event} considers $\mathcal{H}_\infty$ performance and proposes an LMI-based sufficiency condition for filter design. Similarly, \cite{Zhang2015} proposes a synthesis procedure guaranteeing closed-loop stability and dissipativity for a type of event-based output feedback systems.
In \cite{YaZhZhYa14}, the problem of distributed state estimation in a sensor network described by a directed graph with communication only between neighbors is considered. 
As in \cite{zhang2015event} and \cite{Zhang2015}, the transmit decision is based on the difference between the actual measurement and the last measurement, which was transmitted. In contrast,
the transmit decision presented herein uses model-based predictions of the output and compares it with the actual measurement, which typically yields more effective triggering decisions (see \cite{TrimpeCampi,SiKeNo14}).

In \cite{MiTiFiDiJoRu12}, local observers combining a Luenberger observer and consensus-like correction are proposed. An LMI-based design is used to synthesize the observer gains according to the periodic-update (full communication) scenario, and, only in a second step, the event-based mechanism is introduced. 
While a similar Luenberger-type observer structure is used herein, 
the closed-loop stability conditions are not based on the periodic communication scenario, but respect the event-based nature of the control system. 

Most of the mentioned references treat the state estimation problem only, while we simultaneously address stability and performance of the state estimation, and stability of the distributed event-based control system that results when local estimates are used for feedback control. \HChange{The developed results generalize to the pure estimation problem; it suffices to set the state feedback gain $F$ (to be made precise below) to zero.}


\emph{Outline:} The distributed event-based estimation and control architecture is presented in Sec.~\ref{Sec:Arch}, and the problem formulation is made precise in Sec.~\ref{Sec:ProbForm}. The closed-loop dynamics are derived in Sec.~\ref{Sec:CLD} and are then used to obtain conditions guaranteeing closed-loop stability in Sec.~\ref{Sec:SA}. The proposed synthesis procedure is introduced in Sec.~\ref{Sec:Perf} and illustrated in simulation examples in Sec.~\ref{Sec:SimExamples}. The article concludes with remarks in Sec.~\ref{Sec:Conclusion}.
\vspace*{-5pt}\section{Architecture}\label{Sec:Arch}
The following section introduces the distributed event-based control system, which is analyzed subsequently. The architecture is similar to \cite{Tr12} and \cite{Sebastian-CDC}.

\vspace*{-5pt}\subsection{Networked Control System}
The following discrete-time linear system is considered
\begin{align}
\begin{split}
x(k)&=A x(k-1) + B u(k-1) + v(k-1)\\
y(k)&=C x(k)+ w(k),
\end{split}\label{eq:LinSys}
\end{align}
where $k$ denotes the time index, $x(k)\in \mathbb{R}^n$ the state at time $k$, $u(k)\in \mathbb{R}^{n_u}$ the input at time $k$, and $y(k) \in \mathbb{R}^p$ the output at time $k$. The disturbances $v$ and $w$ are bounded \HChange{(but not necessarily deterministic)}, $(A,B)$ is assumed to be stabilizable, and $(A,C)$ is assumed to be detectable.

The inputs and outputs of the system are measured by independent sensor-actuator agents. Therefore the input $u$ and output $y$ is split up according to 
\vspace{-6pt}
\begin{align}
&B \, u(k-1)
= \begin{bmatrix}
B_1 \! & \! B_2 \! & \! \dots \! & \! B_N 
\end{bmatrix}
\begin{bmatrix}
u_1(k-1)\\[-2mm]
\footnotesize\vdots\normalsize \\
u_N(k-1)
\end{bmatrix}
\displaybreak[0] \\
&y(k) =
\begin{bmatrix}
y_1(k) \\[-2mm]
\footnotesize\vdots\normalsize \\
y_N(k)
\end{bmatrix}
= 
\begin{bmatrix}
C_1 \\[-2mm]
\footnotesize\vdots\normalsize \\
C_N
\end{bmatrix}
x(k) +
\begin{bmatrix}
w_1(k) \\[-2mm]
\footnotesize\vdots\normalsize \\
w_N(k)
\end{bmatrix},
\label{eq:system_ys}
\end{align}
where $u_i(k) \in \mathbb{R}^{q_i}$ is agent $i$'s input and $y_i(k) \in \mathbb{R}^{p_i}$ its measurement. The agents can be heterogeneous, thus the dimensions $q_i$ and $p_i$ may differ, including the cases $q_i=0$ and $p_i=0$. It is \emph{not} assumed that the system is detectable or stabilizable by a single agent, i.e. $(A,B_i)$ is not necessarily stabilizable and $(A,C_i)$ not necessarily detectable.

The agents can exchange sensor data $y_i(k)$ with each other
over a broadcast network; that is, if one agent communicates,
all other agents will receive the data.  The communication is assumed to be instantaneous and the agents are synchronized in time. The event-based
mechanism determining when sensor data is exchanged will
be made precise in the next subsection.  
It is assumed that the network bandwidth is sufficient to support such communication, and contention among the agents is resolved by low-level protocols. 
In the Controller Area Network (CAN) on the Balancing Cube \cite{trimpe2012balancing}, for example, contention is resolved through fixed priorities, and the network bandwidth is sufficient to support communication of several agents in one time step.
In contrast to \cite{Tr12,Sebastian-CDC} the agents do not share input data $u_i(k)$ among each other.

We assume that a static state-feedback controller 
$u(k)=F x(k)$ 
is given, rendering $A+BF$ asymptotically stable (all eigenvalues lie strictly within the unit circle). \HChange{The existence of such a feedback gain is guaranteed since $(A,B)$ is stabilizable.} The controller can be designed using standard methods, see e.g. \cite{AsWi97}.

\vspace{-2pt}
\subsection{Distributed Event-Based State Estimation}\label{Sec:EBProtocol}
\vspace{-2pt}
Each agent implements an \emph{event generator} that makes the transmit decision for the local measurement, and a \emph{state estimator} that computes a local state estimate.

\subsubsection{Event Generator}
The event generator triggers the communication of a local measurement $y_i(k)$ of agent $i$ to all other agents. The transmit decision is made according to
\begin{equation}
\text{transmit $y_i(k)$} 
\Leftrightarrow 
|\Delta_i^{-1} \left(y_i(k) - C_i \hat{x}_i(k|k-1)\right)| \geq 1,
\label{eq:eventTrigger}
\end{equation}
where $\Delta_i \in \mathbb{R}^{p_i \times p_i}$ is symmetric and positive definite, $\hat{x}_i(k|k-1)$ is agent $i$'s prediction of the state $x(k)$ based on measurements until time $k-1$ (to be made precise below), 
$C_i \hat{x}_i(k|k-1)$ is agent $i$'s prediction of its measurement $y_i(k)$, and the Euclidean norm is denoted by $|\cdot |$. The communication thresholds $\Delta_i$ will enter the design process as decision variables. 

The underlying idea of the trigger \eqref{eq:eventTrigger} is that a communication should happen whenever the predicted output does not match the actual measurement $y_i(k)$. Such triggers have been considered under the terms \emph{measurement-based trigger}, \emph{innovation-based trigger}, or \emph{predictive sampling} in \cite{TrDAn11,WuJiJoSh13,SiKeNo14,TrimpeCampi}, for example.

To simplify notation, the index set of all agents transmitting their measurements at time $k$ is denoted by 
\ifsingleColumn
\begin{equation}
I(k) := \left\{ i \in \mathbb{N} \; | \; 1 \leq i \leq N, 
|\Delta_i^{-1} \left(y_i(k) - C_i \hat{x}_i(k|k-1)\right)| \geq 1 \right\},
\label{eq:I}
\end{equation}
\else
\begin{align}
\begin{split}
I(k) := &\left\{ i \in \mathbb{N} \; | \; 1 \leq i \leq N, \right. \\
&\hspace{1cm}\left.  |\Delta_i^{-1} \left(y_i(k) - C_i \hat{x}_i(k|k-1)\right)| \geq 1 \right\},
\end{split}
\label{eq:I}
\end{align}
\fi
where $\mathbb{N}$ denotes the set of natural numbers.

\subsubsection{State Estimator}
Each agent estimates the full state $x$. Let $\hat{x}_i(k)=\hat{x}_i(k|k)$ denote agent $i$'s estimate of the state at time $k$ given measurement data up to time $k$, which is computed by
\ifsingleColumn
\begin{align}
\hat{x}_i(k|k-1) &= A  \hat{x}_i(k-1|k-1) + B \hat{u}^i(k-1) \label{eq:EBSE1}  \\
\hat{x}_i(k) &= \hat{x}_i(k|k-1)
+ \sum_{j \in I(k)} L_j \big( y_j(k) - C_j \hat{x}_i(k|k-1) \big) +d_i(k),   \label{eq:EBSE2}
\end{align}
\else
\begin{align}
\hat{x}_i(k|k-1) &= A  \hat{x}_i(k-1|k-1) + B \hat{u}^i(k-1) \label{eq:EBSE1}  \\
\hat{x}_i(k) &= \hat{x}_i(k|k-1) \label{eq:EBSE2} \\
&\phantom{=}+ \sum_{j \in I(k)} L_j \big( y_j(k) - C_j \hat{x}_i(k|k-1) \big) +d_i(k),  \nonumber
\end{align}
\fi
where $\hat{u}^i(k)$ is agent $i$'s belief of the input $u(k)$, $L_j$ are observer gains to be designed, and $d_i$ represents a disturbance, which is assumed to be bounded. 
The disturbance $d_i$ models\footnote{\HChange{We emphasize that $d_i$ is introduced as a generic disturbance signal for the purpose of stability analysis. When implementing the event-based estimator \eqref{eq:EBSE1}, \eqref{eq:EBSE2}, $d_i(k)$ is omitted.}} mismatches between the estimates of the individual agents, which may stem from unequal initialization, different computation accuracy, or imperfect communication. \HChange{For example, if the communication from agent $m$ to agent $i$ fails at time $k$, the disturbance $d_i(k)$ takes the value}
\begin{equation}
d_i(k)= - L_m (y_m(k)-C_m \hat{x}_i(k|k-1)). \label{eq:di1}
\end{equation}
\HChange{In Sec.~\ref{Sec:SimExamples}, random packet drops are simulated in this way. While $d_i$ cannot be bounded for random drops in general, the simulation results demonstrate that the design is effective also in this case. In 
\ifArxiv
App.~\ref{App:ModPacketDrops},
\else
\cite[App.~D]{extendedVersion},
\fi
we discuss a packet drop model where the assumption of bounded disturbances is valid provided that packet drops are sufficiently rare.}

The disturbance signal $d_i$ in \eqref{eq:EBSE2}, which may cause the agents' estimates to differ, plays a crucial role with regards to stability. 
While closed-loop stability is shown in \cite{Tr12} for $d_i = 0$, it was found in \cite{Sebastian-CDC}
that stability can be lost in case $d_i \neq 0$. To recover
stability even in case of nonzero disturbances $d_i$, periodic
estimator resets were introduced in \cite{Sebastian-CDC}. 
By incorporating the event-based and distributed nature of the control system in the observer design herein, the communication of inputs and (under favorable circumstances) the periodic estimator resets are avoided, while still guaranteeing closed-loop stability for $d_i\neq 0$.

The communication protocol \eqref{eq:eventTrigger} implies that a measurement is either transmitted to all agents (and thus included in all state estimates \eqref{eq:EBSE2}), or it is discarded. \HChange{In
\ifArxiv
App.~\ref{Subsec:LocMeasUp},
\else
\cite[App.~C]{extendedVersion},
\fi
the case where each agent updates its state estimate with its local measurements $y_i$ at every time step is discussed. It is shown that stability is still preserved, while at the same time the estimation performance might be improved.}


\subsubsection{Distributed Control}
Given agent $i$'s state estimate, its local input $u_i$ is obtained by
\begin{equation}
u_i(k)=F_i \hat{x}_i(k), \label{eq:TrueFB}
\end{equation}
where $F\T=[F_1\T, F_2\T,\dots,F_N\T]$ is the decomposition of the feedback gain $F$ according to the dimensions of $u_1(k), u_2(k), \dots, u_N(k)$. Agent $i$'s belief $\hat{u}^i(k)$ of the complete input $u(k)$ is defined as
\begin{equation}
\hat{u}^i(k):=F \hat{x}_i(k), \label{eq:EstFB}
\end{equation}
and is used in the state estimator update \eqref{eq:EBSE1}. This contrasts earlier work, \cite{Tr12,Sebastian-CDC}, where it was assumed that each agent has access to the true input $u(k)$. Hence, we do not require the communication of the inputs $u_i(k)$ in this work, which reduces the communication load.

\vspace*{-3pt}\section{Problem Formulation}\label{Sec:ProbForm}
The objective of this article is to present a synthesis procedure for both the estimator gains $L_i$ and the communication thresholds $\Delta_i$. \HChange{The estimators are designed to guarantee i) closed-loop stability (stable dynamics \eqref{eq:LinSys}, \eqref{eq:EBSE1}, \eqref{eq:EBSE2}, \eqref{eq:TrueFB}, and \eqref{eq:EstFB} for bounded disturbances $v$, $w_i$, and $d_i$), and ii) achieve a predefined $\mathcal{H}_2$ performance incorporating estimation and communication objectives.}

\vspace*{-3pt}\section{Closed-loop Dynamics}\label{Sec:CLD}
In this section, the closed-loop dynamics are expressed in terms of the system state, local estimation errors, and inter-agent estimation errors, which forms the basis for deriving the stability conditions in Sec.~\ref{Sec:SA}.
\HChange{This decomposes the closed-loop dynamics into a series of subsystems connected in 
feedforward, which facilitates the subsequent analysis. We obtain
}
\begin{align}
x(k)=&(A\!+\!BF)x(k-1) -\sum_{i=1}^N B_i F_i e_i(k-1)+v(k-1),
\label{eq:CLx}
\end{align}
where $e_i$ is the estimation error of agent $i$ defined by $e_i:=x-\hat{x}_i$,
\begin{align}
\label{eq:AgentErrorDyn}
\begin{split}
e_i&(k)=(I-LC) A e_i(k-1) + (I-LC) v(k-1)\\
&+(I-LC) \sum_{j=1}^N B_j F_j \epsilon_{ji}(k-1) +\xi(k) - d_i(k)\\
&+\sum_{j\in \comp{I}(k)} L_j C_j (A+BF)\epsilon_{ji}(k-1) -\sum_{j=1}^N L_j w_j(k)
\end{split}
\end{align}
with 
\begin{equation}
\xi(k):=\sum_{j\in \comp{I}(k)}L_j (y_j(k)-C_j \hat x_j(k|k-1)),
\label{eq:Deltai}
\end{equation}
where $\comp{I}(k)$ denotes the complement of $I(k)$ and $\epsilon_{ji}:=\hat{x}_j-\hat{x}_i$ refers to the inter-agent error, and 
\begin{equation}
\epsilon_{ji}(k)=A_\text{cl}(I(k))\epsilon_{ji}(k-1)+d_{ji}(k), \label{eq:InterAgentErr}
\end{equation}
with $d_{ji}$ defined as $d_{ji}:=d_j-d_i$, and
\begin{equation}
A_\text{cl}(I(k)):=(I-\sum_{m\in I(k)}L_m C_m)(A+BF). \label{eq:DefAcl}
\end{equation}

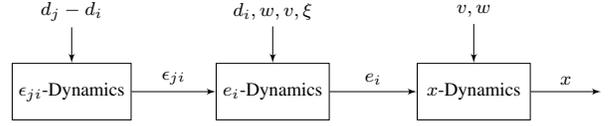
\begin{figure}
\centering
\ifsingleColumn
\resizebox{10.5cm}{!}{
\tikzstyle{int}=[draw, minimum size=2em]
\tikzstyle{init} = [pin edge={to-,thin,black}]
\tikzstyle{sum} = [draw, circle, node distance=1cm]
\tikzstyle{block} = [draw, rectangle, 
    minimum height=3em, minimum width=6em]

\begin{tikzpicture}[node distance=2.5cm,auto,>=latex']
	\node [coordinate] (start){};
	\node [above of=start, node distance=.4cm](dij){$d_j-d_i$};
	\node [block,below of=dij, node distance=1.2cm](eij){$\epsilon_{ji}$-Dynamics};
	\draw [->] (dij) -- (eij);
	\node [block, right of=eij, node distance=3cm](ei){$e_i$-Dynamics};
	\draw [->] (eij) -- node {$\epsilon_{ji}$} (ei);
	\node [above of=ei, node distance=1.2cm](dist1){$d_i, w, v, \xi$};
	\draw [->] (dist1) -- (ei);
	\node [block, right of=ei, node distance=3cm](x){$x$-Dynamics};
	\draw [->] (ei) -- node {$e_i$} (x);
	\node [above of=x, node distance=1.2cm](dist2){$v, w$};
	\draw [->] (dist2) -- (x);
	\node [right of=x, node distance=2cm](end){};
	\draw [->] (x) -- node {$x$} (end);
\end{tikzpicture}}
\else
\resizebox{8.5cm}{!}{
\tikzstyle{int}=[draw, minimum size=2em]
\tikzstyle{init} = [pin edge={to-,thin,black}]
\tikzstyle{sum} = [draw, circle, node distance=1cm]
\tikzstyle{block} = [draw, rectangle, 
    minimum height=3em, minimum width=6em]

\begin{tikzpicture}[node distance=2.5cm,auto,>=latex']
	\node [coordinate] (start){};
	\node [above of=start, node distance=.4cm](dij){$d_j-d_i$};
	\node [block,below of=dij, node distance=1.2cm](eij){$\epsilon_{ji}$-Dynamics};
	\draw [->] (dij) -- (eij);
	\node [block, right of=eij, node distance=3cm](ei){$e_i$-Dynamics};
	\draw [->] (eij) -- node {$\epsilon_{ji}$} (ei);
	\node [above of=ei, node distance=1.2cm](dist1){$d_i, w, v, \xi$};
	\draw [->] (dist1) -- (ei);
	\node [block, right of=ei, node distance=3cm](x){$x$-Dynamics};
	\draw [->] (ei) -- node {$e_i$} (x);
	\node [above of=x, node distance=1.2cm](dist2){$v, w$};
	\draw [->] (dist2) -- (x);
	\node [right of=x, node distance=2cm](end){};
	\draw [->] (x) -- node {$x$} (end);
\end{tikzpicture}}
\vspace*{-2pt}
\fi
\caption{\HChange{Simplified block diagram representing the closed-loop system as a feedforward connection of subsystems.} The disturbances $d_i$, $w$, $v$, and $\xi$ are bounded (either by assumption or by the event-triggering rule \eqref{eq:eventTrigger}).
}
\vspace*{-11pt}
\label{fig:strictFF}
\end{figure}

\ifsingleColumn
\else
\vspace*{-8pt}
\fi
\section{Stability Analysis}\label{Sec:SA}
Next, conditions on the observer gains $L_i$ are derived to guarantee stability of the closed-loop system. These conditions are expressed as LMIs and can be used for the synthesis of stabilizing observer gains $L_i$ as presented in Sec.~\ref{Sec:Perf}.

\HChange{Stability is discussed using the concept of input-to-state stability (ISS) as defined in \cite[Def.~3.1]{jiang2001input}. A feedforward connection of systems is ISS if each system is ISS by itself \cite[Cor.~1]{ChangingSupplyFunctions}. Since this applies to the closed-loop dynamics (see Fig.~\ref{fig:strictFF}), conditions guaranteeing ISS for each subsystem (i.e. the inter-agent dynamics \eqref{eq:InterAgentErr}, the agent error \eqref{eq:AgentErrorDyn}, and the system state \eqref{eq:CLx}) are derived first to subsequently conclude stability for the entire system.}

\subsection{Stability of the Inter-Agent Error}
For the subsequent analysis, the inter-agent error \eqref{eq:InterAgentErr} is regarded as a switched linear system under arbitrary switching. \HChange{While the event-based design will not typically lead to arbitrary switching, it is difficult to determine all possible communication patterns a-priori, without additional restrictions on the system's structure and the disturbances. However, the consideration of arbitrary switching provides a means to derive general stability conditions that can be expressed as LMIs.}
The following theorem establishes stability of the inter-agent error dynamics \eqref{eq:InterAgentErr} by means of a switched quadratic Lyapunov function. This result extends the one in previous work \cite{muehlebach2015lmi}, which employed a common Lyapunov function leading to a more conservative condition. 

\begin{theorem}
Let the matrix inequalities
\begin{align}
\begin{split}
A_\textnormal{cl}\T(\Pi_i) P_1 A_\textnormal{cl}(\Pi_i) - P_1 < 0, \quad &A_\textnormal{cl}\T(\Pi_i) P_1 A_\textnormal{cl}(\Pi_i) - P_2 < 0,\\
A_\textnormal{cl}\T(\emptyset) P_2 A_\textnormal{cl}(\emptyset) - P_2 < 0, \quad 
&A_\textnormal{cl}\T(\emptyset) P_2 A_\textnormal{cl}(\emptyset) - P_1 <0,
\end{split}\label{eq:IESCond1}
\end{align}
be fulfilled for symmetric positive definite matrices $P_1, P_2 \in \mathbb{R}^{n\times n}$, and for all $\Pi_i \in \Pi \setminus \emptyset$, where $\emptyset$ denotes the empty set and $\Pi$ the power set of $\{1,2,\dots,N\}$.
Then the inter-agent error \eqref{eq:InterAgentErr} is ISS. \label{Thm:IES1}
\end{theorem}
\begin{proof}
Consider a trajectory $\epsilon_{ji}(k)$, $k=1,2,\dots$, subjected to \eqref{eq:InterAgentErr} and starting at $\epsilon_{ji}(0)$. Let the trajectory $V$ be defined as
\begin{equation}
V(k)=\begin{cases} \epsilon_{ji}\T(k) P_1 \epsilon_{ji}(k) &I(k)\neq \emptyset\\
\epsilon_{ji}\T(k) P_2 \epsilon_{ji}(k) &I(k)=\emptyset,
\end{cases}\label{eq:LyapFun}
\end{equation}
$k=0,1,\dots$.
Note that $V(k) \geq 0$ for all $k$, where equality holds only if $\epsilon_{ji}(k)$ vanishes. Moreover, $V$ can be bounded by
\begin{align}
0 \leq & \munderbar\sigma |\epsilon_{ji}(k)|^2 \leq V(k) \leq \bar{\sigma} |\epsilon_{ji}(k)|^2,
\end{align}
where $\munderbar\sigma:=\min\{\sigma_{\text{min}}(P_1), \sigma_{\text{min}}(P_2)\}$ and $\bar{\sigma}:=\max\{\sigma_{\text{max}}(P_1), \sigma_{\text{max}}(P_2)\}$, and $\sigma_{\text{min}}(P)$, $\sigma_{\text{max}}(P)$ denote the minimum and maximum singular values of a matrix $P$.
The time evolution of $V$ is given by
\begin{align*}
V(k)&-V(k-1)=
2 d_{ji}\T(k) P_m A_{\text{cl}}(I(k)) \epsilon_{ji}(k-1)\\
&+\epsilon_{ji}\T(k-1)\left(A_\text{cl}\T(I(k)) P_m A_\text{cl}(I(k)) - P_l\right) \epsilon_{ji}(k-1)\\
&+d_{ji}\T(k) P_m d_{ji}(k),
\end{align*}
where $m \in\{1,2\}$, $l\in \{1,2\}$, depending on $I(k)$ and $I(k-1)$. Denoting the maximum eigenvalue of $A_\text{cl}\T(\Pi_i) P_m A_\text{cl}(\Pi_i) - P_l$ over all $\Pi_i \in \Pi$ by $\bar{\lambda}$, yields the bound
\ifsingleColumn
\begin{equation*}
V(k)-V(k-1)\leq 2 |d_{ji}(k)| |P_m A_\text{cl}(I(k))| |\epsilon_{ji}(k-1)|
+\bar{\lambda} |\epsilon_{ji}(k-1)|^2 + |P_m| |d_{ji}(k)|^2.
\end{equation*}
\else
\begin{multline*}
V(k)-V(k-1)\leq 2 |d_{ji}(k)| |P_m A_\text{cl}(I(k))| |\epsilon_{ji}(k-1)|\\
+\bar{\lambda} |\epsilon_{ji}(k-1)|^2 + |P_m| |d_{ji}(k)|^2.
\end{multline*}
\fi
Completing the squares with an $\alpha>0$ results in
\ifsingleColumn
\begin{align*}
V(k)-V(k-1)\leq (\bar{\lambda}+\alpha)& |\epsilon_{ji}(k-1)|^2 \!\!-\left(\sqrt{\alpha} |\epsilon_{ji}(k-1)| - \frac{|P_m A_{\text{cl}}(I(k))|}{\sqrt{\alpha}} |d_{ji}(k)|\right)^2\\
&\!\!+\left(  \frac{|P_m A_{\text{cl}}(I(k))|^2}{\alpha} + |P_m|\right) |d_{ji}(k)|^2.
\end{align*}
\else
\begin{align*}
V&(k)-V(k-1)\leq (\bar{\lambda}+\alpha) |\epsilon_{ji}(k-1)|^2 \\
&\!\!-\left(\sqrt{\alpha} |\epsilon_{ji}(k-1)| - \frac{|P_m A_{\text{cl}}(I(k))|}{\sqrt{\alpha}} |d_{ji}(k)|\right)^2\\
&\!\!+\left(  \frac{|P_m A_{\text{cl}}(I(k))|^2}{\alpha} + |P_m|\right) |d_{ji}(k)|^2.
\end{align*}
\fi
Therefore
\ifsingleColumn
\begin{align*}
V(k)-V(k-1)\leq (\bar{\lambda}+\alpha) |\epsilon_{ji}(k-1)|^2 
+ \left(  \frac{|P_m A_{\text{cl}}(I(k))|^2}{\alpha} + |P_m|\right) |d_{ji}(k)|^2
\end{align*}
\else
\begin{align*}
V(k)&-V(k-1)\leq (\bar{\lambda}+\alpha) |\epsilon_{ji}(k-1)|^2 \\
&+ \left(  \frac{|P_m A_{\text{cl}}(I(k))|^2}{\alpha} + |P_m|\right) |d_{ji}(k)|^2
\end{align*}
\fi
and consequently
\begin{align}
V(k)&\leq a V(k-1)+b |d_{ji}(k)|^2, \label{eq:evalV}
\end{align}
where 
\begin{align*}
b&:=\max_{\Pi_i \in \Pi, m\in\{1,2\}} \left(  \frac{|P_m A_{\text{cl}}(\Pi_i)|^2}{\alpha} + |P_m|\right),
a:=\frac{\bar{\lambda}+\alpha}{\munderbar\sigma}+1.
\end{align*}
By assumption, c.f. \eqref{eq:IESCond1}, $\bar{\lambda}$ is negative and therefore an $\alpha>0$ can be chosen such that $0<a<1$. As a consequence, \eqref{eq:evalV} implies that $V(k)$ remains bounded for all $k$. In particular, it follows that
\begin{align}
V(k) \leq a^k V(0) + b \sum_{l=0}^{k-1} a^l |d_{ji}(k-l)|^2,
\end{align}
and therefore
\begin{align}
|\epsilon_{ji}(k)|^2 \leq a^k \frac{\bar{\sigma}}{\munderbar \sigma} |\epsilon_{ji}(0)|^2 + \frac{b}{\munderbar \sigma} \sum_{l=0}^{k-1} a^l |d_{ji}(k-l)|^2.
\end{align}
The constants $\munderbar \sigma, \bar{\sigma}, a, b$ are all positive, which results in
\begin{equation}
|\epsilon_{ji}(k)| \leq a^{\frac{k}{2}} \sqrt{\frac{\bar{\sigma}}{\munderbar \sigma}} |\epsilon_{ji}(0)| + \sqrt{\frac{b}{\munderbar \sigma}} \sum_{l=0}^{k-1} a^{\frac{l}{2}} |d_{ji}(k-l)|, \label{eq:ejidecay}
\end{equation}
and proves that the inter-agent error is ISS.
\end{proof}

In Thm.~\ref{Thm:IES1}, the Lyapunov function is switched depending on whether there is communication or not. Using the Schur complement, the conditions \eqref{eq:IESCond1} can be rewritten as
\begin{equation}
\left( \begin{array}{cc} P_k & P_k (I-\sum_{m\in \Pi_i} L_m C_m)(A+BF)\\
* &P_l\end{array} \right)>0, \label{eq:SchurComp}
\end{equation}
for all $\Pi_i \in \Pi$ with $k=1$ if $\emptyset\not\in \Pi_i$, $k=2$ if $\Pi_i=\{\emptyset\}$ and $l=1,2$\HChange{, where the placeholder $*$ is implied by symmetry of the matrix}. Thus, using the change of variables $U_m=P_1 L_m$, the previous set of matrix inequalities is linear in $U_m$, $P_1$, and $P_2$ for all $m=1,2,\dots,N$ and can therefore be used as auxiliary condition for the synthesis of the observer gains $L_m$, as done in Sec.~\ref{Sec:Perf}.

By introducing a Lyapunov function that switches for each communication pattern (i.e. distinct $P_i$'s for each $\Pi_i\in \Pi$), and not only between the case of communication or no communication, the conservativeness of Thm.~\ref{Thm:IES1} could be reduced further. However, in that case the resulting stability conditions are not suitable for synthesis, as they are no longer linear in the decision variables. In addition, such an extension would result in a significant increase in the number of LMIs (number of LMIs of the order $2^{2N}$).

\subsubsection{Relaxation of the LMI conditions in Thm.~\ref{Thm:IES1}}
In the following, we aim  to reduce the number of LMI conditions required to guarantee inter-agent error stability. We first note that the result from \cite{muehlebach2015lmi} follows from Thm.~\ref{Thm:IES1} as a corollary,
\begin{corollary}
Let the matrix inequality
\begin{equation}
A_\textnormal{cl}\T(\Pi_i) P A_\textnormal{cl}(\Pi_i) - P < 0, \label{eq:IEScond1}
\end{equation}
be satisfied for a symmetric positive definite matrix $P\in \mathbb{R}^{n\times n}$ and for all $\Pi_i\in \Pi$, where $\Pi$ denotes the power set of $\{1,2,\dots,N\}$. Then the inter-agent error is ISS.\label{Cor:IES2}
\end{corollary}
\begin{proof}
Set $P_1=P_2$ in Thm.~\ref{Thm:IES1}.
\end{proof}
\vspace*{-5pt}
%
The power set $\Pi$ has cardinality $2^N$, which leads to a rapid growth in the number of LMIs used to ensure inter-agent stability even in Cor.~\ref{Cor:IES2}. For a large number of agents, the corresponding synthesis problem may become intractable. Therefore the conditions from Cor.~\ref{Cor:IES2} are further relaxed, such that the number of LMIs scales linearly with the number of agents. This comes at the price of more conservative conditions.

\begin{corollary}\label{Cor:IES3}
Let the matrix inequalities
\begin{align}
H \geq \left( \begin{array}{cc} P &P(A+BF)\\ (A+BF)\T P &P\end{array}\right) &> 0,\label{eq:cond1}\\
\left( \begin{array}{cc} P &P (I-L_m C_m)(A+BF) \\
* &P \end{array} \right) &> \frac{N-1}{N} H \label{eq:cond2}
\end{align}
be satisfied for a symmetric positive definite matrix $H \in \mathbb{R}^{2n \times 2n}$, a symmetric positive definite matrix $P \in \mathbb{R}^{n \times n}$, and for all $m \in \{1,2,\dots,N\}$. Then the inter-agent error is ISS.
\end{corollary}
\begin{proof}
Applying the Schur complement to \eqref{eq:IEScond1} results in \begin{equation}
\left( \begin{array}{cc} P & P A_\text{cl}(\Pi_i)\\
A_\text{cl}(\Pi_i)\T P & P \end{array} \right)>0, \label{eq:IEScond1Schur}
\end{equation}
for all $\Pi_i \in \Pi$ and therefore \eqref{eq:cond1} implies \eqref{eq:IEScond1} for $\Pi_i = \emptyset$.
Note that the sum in the expression of $A_\text{cl}(\Pi_i)$ can be rearranged to
\begin{align*}
A_\text{cl}&(\Pi_i)\!=\!-(|\Pi_i|\!-\!1) (A\!+\!BF) + \!\sum_{m \in \Pi_i} \!(I \!- \!L_m C_m) (A\!+\!BF),
\end{align*}
such that the LMI \eqref{eq:IEScond1Schur} can be reformulated as
\ifsingleColumn
\begin{equation}
\sum_{m\in \Pi_i} \left( \begin{array}{cc} P & P(I-L_m C_m)(A+BF) \\
* & P\end{array} \right) > 
(|\Pi_i|-1) \left( \begin{array}{cc} P & P(A+BF)\\
* & P \end{array}\right), \label{eq:tmp123}
\end{equation}
\else
\begin{multline}
\sum_{m\in \Pi_i} \left( \begin{array}{cc} P & P(I-L_m C_m)(A+BF) \\
* & P\end{array} \right) > \\
(|\Pi_i|-1) \left( \begin{array}{cc} P & P(A+BF)\\
* & P \end{array}\right), \label{eq:tmp123}
\end{multline}
\fi
for all $\Pi_i \in \Pi$. In contrast, combining \eqref{eq:cond1} and \eqref{eq:cond2} leads to
\ifsingleColumn
\begin{equation*}
\sum_{m\in \Pi_i} \left( \begin{array}{cc} P & P(I-L_m C_m)(A+BF) \\
* & P\end{array} \right) > 
\frac{N-1}{N} |\Pi_i| \left( \begin{array}{cc} P &P(A+BF)\\ (A+BF)\T P &P\end{array}\right)
\end{equation*}
\else
\begin{multline*}
\sum_{m\in \Pi_i} \left( \begin{array}{cc} P & P(I-L_m C_m)(A+BF) \\
* & P\end{array} \right) > \\ 
\frac{N-1}{N} |\Pi_i| \left( \begin{array}{cc} P &P(A+BF)\\ (A+BF)\T P &P\end{array}\right)
\end{multline*}
\fi
for all $\Pi_i \in \Pi \setminus \emptyset$.
It holds that $|\Pi_i| (N-1)/N \geq (|\Pi_i|-1)$, and therefore \eqref{eq:cond1} and \eqref{eq:cond2} imply \eqref{eq:tmp123} (and thereby also \eqref{eq:IEScond1}) for all $\Pi_i \in \Pi \setminus \emptyset$, which concludes the proof.
\end{proof}
In Sec.~\ref{Sec:SimExamples}, the different stability conditions are compared by means of simulation examples.

\begin{remark}
In case the open-loop system is unstable, it is essential for guaranteeing inter-agent error stability that each agent reconstructs the input $u$ based on its current state estimate $\hat{x}_i$, as opposed to the case where all agents have access to the true input $u$ (proposed in \cite{Sebastian-CDC,Tr12}). This seems counterintuitive, as providing the agents with more information should potentially improve the closed-loop performance. The mechanism leading to a destabilization is further discussed and illustrated on a simple example in
\ifArxiv
App.~\ref{App:ComIn}.
\else
\cite[App.~F]{extendedVersion}.
\fi
\end{remark}

\vspace*{-3pt}\subsection{Stability of the Agent Error}
 
Stability of the agent error \eqref{eq:AgentErrorDyn} follows directly from the agent-error dynamics \eqref{eq:AgentErrorDyn}, the inter-agent error being bounded, and the communication protocol, which bounds the disturbance $\xi$. 
\begin{lemma}
Let the inter-agent errors $\epsilon_{ji}$, $j=1,2,\dots,N$ be bounded. Then the agent error $e_i$ is ISS if and only if the eigenvalues of $(I-LC)A$ have magnitude strictly less than one.
\label{Lem:ES}
\end{lemma}
\begin{proof}
See 
\ifArxiv
App.~\ref{App:LeES}.
\else
\cite[App.~A]{extendedVersion}.
\fi
\end{proof}

We remark that $(I-LC)A$ corresponds to the error dynamics for the estimator \eqref{eq:EBSE1}, \eqref{eq:EBSE2} with full communication; that is, stability of $(I-LC)A$ is a natural requirement for the estimator design. \HChange{Due to the detectability of $(A,C)$, the existence of such estimator gains $L$ is guaranteed.}

\subsection{Stability of the Closed-loop System}
By combining the previous results, conditions for the closed-loop dynamics to be ISS can be established. Provided that the agent error is bounded, it follows from \eqref{eq:CLx} that the state $x$ is ISS, since, by assumption, $A+BF$ has all eigenvalues strictly within the unit circle. This leads to the following conclusion:

\begin{theorem}\label{Thm:CL}
Let the eigenvalues of $A+BF$ have magnitude strictly less than one. The closed-loop system is ISS if both, the agent error \eqref{eq:AgentErrorDyn} and the inter-agent error  \eqref{eq:InterAgentErr} are ISS.
\end{theorem}

\section{Performance Analysis and Synthesis}\label{Sec:Perf}

In this section, 
a general $\mathcal{H}_2$ performance measure is introduced that can capture both estimation performance and communication requirements. LMI-conditions will be established guaranteeing a worst-case performance. 
Moreover, a unified synthesis procedure for the distributed and event-based estimator \eqref{eq:EBSE1}, \eqref{eq:EBSE2} is presented that combines stability requirements (from Sec.~\ref{Sec:SA}) and performance criteria.

\HChange{We will focus on the design of the estimator gains $L_i$ and the communication thresholds $\Delta_i$. However, a similar approach could be used to synthesize the feedback gain $F$ subject to the stability conditions provided by Thm.~\ref{Thm:IES1}, Cor.~\ref{Cor:IES2}, or Cor.~\ref{Cor:IES3}. Likewise, $\mathcal{H}_2$ or $\mathcal{H}_\infty$ performance measures could be included in the design. The resulting synthesis procedures are very similar to the ones presented herein 
and thus not discussed in detail.
}

\vspace*{-3pt}\subsection{Performance Measure}
To simplify the derivation of the performance metric, 
we assume  that the disturbances $d_i$ are absent and that all agents are initialized with the same state estimate. \HChange{According to \eqref{eq:InterAgentErr}, this implies $\epsilon_{ji}(k)=0$ for all $k$
(i.e., all agents' estimates are identical), and as a result, we formulate a performance metric based on the estimation error $e_i$ of a single agent.
We emphasize that this simplification only serves to obtain a tractable performance criterion;  
the final synthesis procedure is then augmented with conditions from the previous section ensuring ISS of the closed-loop system and account for the general case of nonzero disturbances $d_i$.}

With the above assumptions, the estimation error \eqref{eq:AgentErrorDyn} simplifies to
\ifsingleColumn
\begin{equation}
e_i(k)=(I-LC)A e_i(k-1) + (I-LC)v(k-1) - L w(k) + \xi(k).
\label{eq:SimpAgentErrorDyn}
\end{equation}
\else
\begin{align}
\begin{split}
e_i(k)=(I-LC)A e_i(k-1)& + (I-LC)v(k-1) \\
&- L w(k) + \xi(k).
\end{split}\label{eq:SimpAgentErrorDyn}
\end{align}
\fi
The disturbance $\xi(k)$, as defined in \eqref{eq:Deltai}, can be reformulated as $\xi(k)=L \Delta s_1(k)$,
where 
\begin{align}
\Delta&:=\diag( \Delta_1, \Delta_2, \dots, \Delta_N) \in \mathbb{R}^{p \times p},\\
s_{1}(k)&:=(s_{11}\T(k),s_{12}\T(k),\dots,s_{1N}\T(k))\T \in \mathbb{R}^{p},\\
s_{1i}(k)&:= \chi_{i\in \comp{I}}(k) q_i(k) \in \mathbb{R}^{p_i},\\
q_i(k)&:= \Delta_i^{-1} (y_i(k)-C_i \hat{x}_i(k|k-1)) \in \mathbb{R}^{p_i},
\end{align}
and $\chi_{i\in \comp{I}}(k)$ denotes the indicator function, that is, $\chi_{i\in \comp{I}}(k)=1$ if $i\in \comp{I}(k)$ and 0 otherwise, for $k\in \mathbb{N}$ and $i=1,2,\dots,N$.
Note that the signal $q$ is directly related to the communication since a transmission is triggered if $|q_i(k)| > 1$. Furthermore, the communication protocol guarantees that $|s_{1i}(k)|$ is strictly less than one. The agent error dynamics \eqref{eq:SimpAgentErrorDyn} can be represented by the block diagram shown in Fig.~\ref{Fig:AgentErr}.

The communication protocol results in nonlinear feedback terms because of the switching behavior of the event triggers.  Therefore, the direct minimization of the performance criterion (to be made precise below) is difficult. 
Instead, we minimize an upper bound, which is obtained by
 considering the worst-case performance with respect to all perturbations $s_{1i}(k)$ with Euclidean norm less than one. This leads to a robust control problem and, as a consequence, the resulting synthesis procedure can be formulated as a convex optimization problem.

\begin{figure}
\graphicspath{{media/}}
\centering
\ifsingleColumn
\def\svgwidth{0.45\columnwidth}
\else
\def\svgwidth{0.55\columnwidth} 
\fi
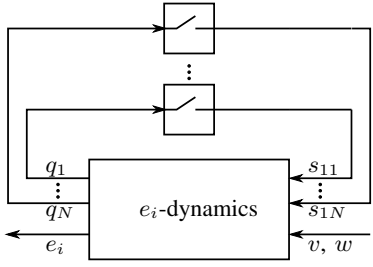
\vspace*{-3pt}
\caption{Block diagram of the simplified agent error dynamics \eqref{eq:SimpAgentErrorDyn}. The error $e_i$ is driven by the external disturbances $v$ and $w$. The switches and the signals $q_i$ and $s_{1i}$ are used to model the event-based communication. Based on the magnitude of the signal $q_i(k)$ at time instant $k$, the $i$th switch is either closed (no communication in case $|q_i(k)|<1$) implying $q_i(k)=s_{1i}(k)$, or opened (communication in case $|q_i(k)|\geq 1$) implying $s_{1i}(k):=0$.}
\label{Fig:AgentErr}
\end{figure}

The power semi-norm \cite[p. 816]{Oppenheim} is used as performance objective:\vspace*{-2pt}
\begin{equation}
\vspace{-11pt}||z||_\mathcal{P}:=\lim_{K\rightarrow \infty} \sqrt{ \frac{1}{K} \sum_{k=1}^K z\T(k) z(k) }, 
\end{equation}
where 
\begin{equation}
z(k):=\hat{C} e_i(k-1) + \hat{D}_{21} w(k) + \hat{D}_{22} v(k-1), \label{eq:DefZ}
\end{equation}
with $\hat{C}$, $\hat{D}_{21}$, $\hat{D}_{22}$ arbitrary matrices of appropriate dimensions. In particular, \eqref{eq:DefZ} allows for the choices 
\ifsingleColumn
$z(k)=(q_1\T(k),\dots,$ $q_N\T(k))\T$
\else
$z(k)=(q_1\T(k),q_2\T(k),\dots,q_N\T(k))\T$
\fi
and $z(k)=e_i(k)$, which can be used to reduce, respectively, average communication and estimation error, as shall be demonstrated later. 

In the following, a synthesis procedure for the observer gains $L_i$ and the communication thresholds $\Delta_i$ is developed, which seeks to minimize $||z||_\mathcal{P}$. 
However, for the reasons stated above, we do not minimize $||z||_\mathcal{P}$ directly, but an upper bound,
which is formulated in terms of $\mathcal{H}_2$ and $\mathcal{H}_\infty$ norms.
Expressing the $\mathcal{H}_2$ and $\mathcal{H}_\infty$ norms using LMIs, see e.g. \cite{de2002extended}, leads to the following result:
\begin{theorem}
Let the disturbances $v(k)$ and $w_i(k)$ be bounded, zero mean, independent and identically distributed for all $k$ with covariances $V$ and $W_i$, respectively, $i=1,2,\dots, N$. 
Define
\small
\begin{align*}
&\hat{A}:=(I-LC)A, \qquad \quad
\quad ~\hat{B}_{2}:=\begin{bmatrix} -LW^{\frac{1}{2}} \! & (I-LC)V^{\frac{1}{2}} \end{bmatrix}, \\
&\hat{D}_{2}:=\begin{bmatrix} \hat{D}_{21} W^{\frac{1}{2}} \! &\hat{D}_{22} V^{\frac{1}{2}}\end{bmatrix}, W:=\diag(W_1, W_2,\dots, W_N),
\end{align*}%
\normalsize
and let the matrix inequalities %
\small
\begin{gather}
\left( \begin{array}{ccc} I &0 &\hat{C}\\
							 	0 &P &P\hat{A}\\
								\hat{C}\T &\hat{A}\T P &P\end{array} \right)>0,
\left( \begin{array}{ccc} I &  0& \hat{D}_{2} \\ 
0 &P &P\hat{B}_{2}\\
\hat{D}_{2}\T &\hat{B}_{2}\T P &X\end{array} \right) >0, \label{eq:H2lmi}\\
\vspace{-30pt}\left( \begin{array}{cccc} Q &\hat{A}Q &L \Delta &0\\
						Q\hat{A}\T &Q &0 &Q\hat{C}^{\mathrm{T}} \\
						\Delta L\T &0 &I &0\\
						0 &\hat{C}Q &0 &\gamma I\end{array} \right)>0,\label{eq:Hinflmi}
\end{gather}
\normalsize
be fulfilled for symmetric matrices $P \in \mathbb{R}^{n\times n}$, $Q \in \mathbb{R}^{n\times n}$, $X \in \mathbb{R}^{(n+p) \times (n+p)}$, and a scalar $\gamma\in \mathbb{R}$. Then it holds that
\begin{equation}
||z||_\mathcal{P} < \sqrt{N \gamma} + \sqrt{\tr{X}}. \label{eq:UB}
\end{equation}\label{Thm:Performance}
\end{theorem}\vspace*{-24pt}
\begin{proof}
See 
\ifArxiv
App.~\ref{App:ThmVI1}.
\else
\cite[App.~B]{extendedVersion}.
\fi
\end{proof}\vspace*{-7pt}
The bound \eqref{eq:UB} consists of two terms: The expression $\sqrt{\tr{X}}$ captures the $\mathcal{H}_2$ gain from the \HChange{disturbances} $v$, $w$ to the signal $z$, whereas the expression $\sqrt{N \gamma}$ captures the $\mathcal{H}_\infty$ gain from the signal $s_1$ to the signal $z$, and bounds as such the effect of the nonlinear feedback due to the event-based communication, see Fig.~\ref{Fig:AgentErr}. In the full communication scenario it holds that $s_{1i}=0$, and therefore the agent error reduces to a linear system excited by the \HChange{disturbances} $v$ and $w$, which implies $||z||_\mathcal{P}< \sqrt{\tr{X}}$. Hence, the term $\sqrt{\tr{X}}$ corresponds to the performance in the full communication case and represents a lower bound on the achievable performance in the event-based scenario, which is attained for $\Delta_i \rightarrow 0$. The term $\sqrt{N \gamma}$ bounds the effect of the disturbance $s_1$ due to the event-based communication.


\vspace*{-6pt}\subsection{Synthesis}\label{Sec:SubsecSyn}

\HChange{We first discuss the synthesis of the estimator gains $L_i$ and the thresholds $\Delta_i$ for the relevant special case where the performance measure is the estimation error, which corresponds to the steady-state Kalman filter objective. We then comment on a synthesis procedure for a general performance measure.}

\subsubsection{Kalman Filter Objective}
In case the performance measure is chosen as $z(k)=e_i(k-1)$, that is $\hat{C}=I$, $\hat{D}_2=0$, and $\hat{D}_3=0$ in \eqref{eq:DefZ}, it follows that \eqref{eq:H2lmi} does not depend on the communication thresholds $\Delta_i$. We therefore propose to design the observer gains $L_i$ in a first step by minimizing $\sqrt{\tr{X}}$ subject to \eqref{eq:H2lmi} and to the conditions ensuring closed-loop stability. For example, if the stability conditions provided by Cor.~\ref{Cor:IES2} are used, we synthesize the observer gains according to
\begin{align}
\begin{split}
&\inf_{X,P,L} \tr{X} \quad \text{subject~to}~P=P\T\\
&\left( \begin{array}{ccc} I &0 &I\\
							 	0 &P &P\hat{A}\\
								I &\hat{A}\T P &P\end{array} \right)>0,
\left( \begin{array}{cc} P &P\hat{B}_{2}\\
\hat{B}_{2}\T P &X\end{array} \right) >0,\\
&\left( \begin{array}{cc} P & P A_\text{cl}(\Pi_i)\\
A_\text{cl}(\Pi_i)\T P & P \end{array} \right)>0,~\forall \Pi_i \in \Pi,
\end{split}\label{eq:KF1}
\end{align}
\HChange{where \eqref{eq:IEScond1} has been rewritten using the Schur complement.}
In the absence of the stability conditions obtained from Cor.~\ref{Cor:IES2}, this optimization would yield a centralized steady-state Kalman filter.  Note that the first condition in \eqref{eq:H2lmi} ensures that $(I-LC)A$ will have all eigenvalues strictly within the unit circle, which implies ISS of the closed-loop system according to Thm.~\ref{Thm:CL}.

\HChange{As a result, the contribution $c^*=\sqrt{\tr{X}}$ to the upper bound given by \eqref{eq:UB} can be calculated, and captures the $\mathcal{H}_2$ gain from the signals $v$ and $w$ to the output $e_i$ in the full communication case.}

\HChange{In a second step, the communication thresholds $\Delta_i$ are synthesized such that an a priori specified worst-case performance $J_\text{max}$ is guaranteed (i.e. $||e_i||_\mathcal{P} < J_\text{max}$). This is achieved by solving}
\begin{align}
&\sup_{Q,\Delta,\gamma} ~\tr{\Delta}
\quad\text{subject~to~$Q=Q\T$~and}\label{eq:ExDesignDelta}\\
&\!\!\left( \begin{array}{cccc} Q &\hat{A}Q &L\Delta &0\\
						Q\hat{A}\T &Q &0 &Q^{\mathrm{T}} \\
						\Delta L\T &0 &I &0\\
						0 &Q &0 &\gamma I\end{array} \right)>0,
\gamma < \frac{1}{N} (J_\text{max}-c^*)^2, \nonumber
\end{align}
\HChange{while keeping the estimator gains $L_i$ fixed. Therefore this two-step procedure has the following interpretation:} In the first step, a lower bound on the achievable cost $||e_i||_\mathcal{P}$ is obtained based on the full communication scenario (i.e. $s_1=0$), while respecting the stability conditions for the inter-agent error. In the second step, the communication thresholds $\Delta_i$ are designed such that the a priori specified worst-case performance $J_\text{max}$ is guaranteed. Hence, the second step can be interpreted as performance versus communication trade-off: increasing $J_\text{max}$ will generally downgrade estimation performance by giving the optimization more flexibility to find larger $\Delta_i$\HChangeN{, which tends to reduce communication.}

\HChange{In general, feasibility of \eqref{eq:KF1} cannot be guaranteed. However, the optimization \eqref{eq:ExDesignDelta} is guaranteed to be feasible provided that $J_\text{max}>c^*$. Details regarding feasibility
and extensions in case \eqref{eq:KF1} is not feasible are discussed in 
\ifArxiv
App.~\ref{App:Feas} and \cite{muehlebach2015lmi}.
\else
\cite[App.~E]{extendedVersion} and \cite{muehlebach2015lmi}.
\fi
}

\subsubsection{General Case}
\HChange{This two-step procedure can also be applied in case of a more general performance objective given by \eqref{eq:DefZ}. The difference is that \eqref{eq:H2lmi} might depend on the communication thresholds $\Delta_i$. As a result, we propose to keep the communication thresholds $\Delta_i$ fixed in the first step, yielding the observer gains $L_i$. In the second step, the observer gains $L_i$ are kept fixed and the communication thresholds are updated by solving an optimization similar to \eqref{eq:ExDesignDelta}. The procedure is then repeated until convergence or satisfactory performance.}

\vspace*{-8pt}\section{Simulation Example}\label{Sec:SimExamples}
\HChange{The presented framework for event-based estimation and control is applied in a simulation example that is based on a simplified model for vehicle platooning. Thereby, the communication versus performance trade-off of the proposed approach is discussed, as well as the scalability with respect to a larger number of agents. Additional simulation studies can be found in 
\ifArxiv
App.~\ref{App:Sim}, \cite{muehlebach2015lmi}, and \cite{muehlebach2015guaranteed}.
\else
\cite[App.~G]{extendedVersion}, \cite{muehlebach2015lmi}, and \cite{muehlebach2015guaranteed}.
\fi}


The problem of vehicle platooning has been studied extensively in the literature, see e.g. \cite{AlamPlatoon}, \cite{1098376}, and references therein. In \cite{Jvanovich}, it is shown that the linear quadratic regulator problem is ill-posed as the number of vehicles tends to infinity. Moreover, \cite{Seiler} shows that string instability occurs for any local linear feedback law, where the input of the $i$th vehicle depends linearly on the relative distance to its two neighbors. This motivates the use of a common network, where the different vehicles can exchange information across the platoon.

Similar to \cite{1098376}, we consider a chain of $M$ vehicles (agents), where each vehicle is modeled as a unit point mass. The aim is to control the velocity and the position of each vehicle relative to its neighbors. \HChange{The following continuous-time model is introduced, c.f. \cite{1098376},}
\begin{equation}
x_i(t):=\!\!\left( \begin{array}{c} p_i(t)\\ r_i(t)-r_{i+1}(t)\end{array}\right), \dot{x}_i(t)=\!\!\left(\begin{array}{c} u_i(t)\\ p_i(t)-p_{i+1}(t)\end{array} \right), 
\end{equation}
$i=1,2,\dots,M-1$, and $x_M(t):=p_M(t)$, $\dot{x}_M(t)=u_M(t)$
with $t\in [0,\infty)$, where $r_i$ and $p_i$ denote the position and velocity of the $i$th vehicle and $u_i$ the normalized force generated by the motor of the $i$th vehicle. The model is discretized with a sampling time of $\unit[20]{ms}$ leading to the model \eqref{eq:LinSys}.\footnote{The same notation is used for continuous and discrete-time signals, e.g. $x(t)$ refers to the continuous-time state trajectory, $x(k)$ to the discrete-time state trajectory.}

Each vehicle measures the distance to the previous vehicle,
except for the first vehicle, which measures its velocity.
The measurements are corrupted by independent, uniformly distributed noise, with $[\unit[-0.1]{m},\unit[0.1]{m}]$ (distance measurements), 
\ifsingleColumn
$[\unitfrac[-0.1]{m}{s},$ $\unitfrac[0.1]{m}{s}]$
\else
$[\unitfrac[-0.1]{m}{s},\unitfrac[0.1]{m}{s}]$
\fi
(velocity measurements). \HChange{Likewise, the inputs $u_i(k)$ are corrupted by independent, uniformly distributed noise 
\ifsingleColumn
$[\unitfrac[-0.01]{m}{s^2},$ $\unitfrac[0.01]{m}{s^2}]$
\else
$[\unitfrac[-0.01]{m}{s^2},\unitfrac[0.01]{m}{s^2}]$
\fi
.} The system is controllable and observable, but neither controllable nor observable for each agent on its own.

A stabilizing feedback controller $F$ is obtained by solving the linear quadratic regulator problem with the identity $I \in \mathbb{R}^{(2M-1) \times (2M-1)}$ and the scaled identity $100 ~I \in \mathbb{R}^{M \times M}$ for weighting the state and input costs.  

\subsubsection{3 Vehicles}
\HChange{We consider first the case of three vehicles ($M=3$). As performance objective, the power of the estimation error, $||e_i||_\mathcal{P}$, is used, and the observer gains $L_i$ and the communication thresholds $\Delta_i$ are designed according to Sec.~\ref{Sec:SubsecSyn}. The optimizations are solved up to a tolerance of $10^{-8}$ using SDPT-3, \cite{Toh99sdpt3}, interfaced through Yalmip, \cite{JohanYALMIP}.}
\HChange{The different stability conditions, that is, the conditions given by Thm.~\ref{Thm:IES1}, Cor.~\ref{Cor:IES2}, and Cor.~\ref{Cor:IES3}, lead in this case to a very similar design of the observer gains. We will therefore focus on the results obtained by Cor.~\ref{Cor:IES2}. However, this does not necessarily need to be the case, as shown in
\ifArxiv
App.~\ref{App:Sim}.
\else
\cite[App.~G]{extendedVersion}.
\fi
For the synthesis of the communication thresholds, $J_\text{max}$ is chosen to be roughly 35 times the power $||e_i||_\mathcal{P}$ corresponding to the full communication case, i.e. $J_\text{max}=0.38$, yielding $\Delta_1=0.107$, $\Delta_2=0.092$, and $\Delta_3=0.106$.}

\HChange{The resulting closed-loop system is studied in simulations, where the first car is initialized with a surplus velocity of $\unitfrac[5]{m}{s}$.} The state estimates of the different agents are initialized with zero. In addition, a communication loss rate of $10\%$ is introduced (independent Bernoulli-distributed). \HChange{The simulation results indicate that the approach is robust also to non-deterministic and potentially unbounded disturbances $d_i$.} Fig.~\ref{Fig:Traj3} shows the evolution of the distances between the three vehicles. \HChange{In steady state, the distance error between the vehicles is kept below $\pm \unit[0.1]{m}$.} The communication rates (smoothed with a moving average filter of length 200) of the different vehicles are depicted in Fig.~\ref{Fig:PerfCar}. The communication rate is normalized such that a rate of $1.0$ corresponds to all agents transmitting their measurements at every time step. \HChange{In steady state, the second vehicle communicates its measurement in around $8\%$ of the time, whereas the first and last vehicle communicate at a rate below $4\%$.}

\HChange{The trade-off between estimation performance and communication is obtained by varying $J_\text{max}$. The corresponding steady-state performance $||e_i||_\mathcal{P}$ and the communication rates of the different designs are evaluated in simulations. Their values were estimated using 20 independent simulations (with different noise realizations) over $\unit[1000]{s}$. The variability among the different noise realizations was found to be negligible and a time horizon of $\unit[1000]{s}$ sufficiently long for transients to be insignificant. The communication versus performance graph, as depicted in Fig.~\ref{Fig:PerfCar} is compared to a centralized discrete-time design with reduced sampling rates.\footnote{\HChange{The centralized design is obtained by re-sampling the discrete-time system \eqref{eq:LinSys} at increasingly lower rates, and then performing a centralized steady-state Kalman filter design based on the performance objective $||e_i||_\mathcal{P}$. The fact that the inputs are also communicated is not accounted for in the corresponding communication rates shown in Fig.~\ref{Fig:PerfCar}.}}
This reveals that a better trade-off is achieved by the event-based design as opposed to the centralized design with reduced periodic sampling rates.}


\begin{figure}
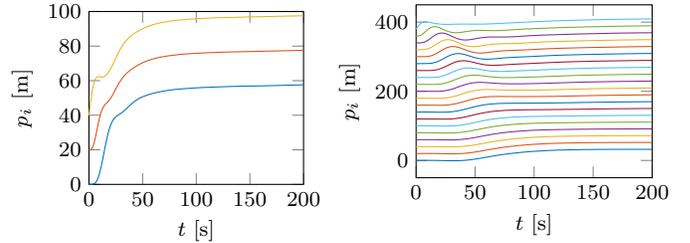

\newlength\figureheight
\newlength\figurewidth
\begin{minipage}[l]{.45\columnwidth}
\ifsingleColumn
\setlength\figureheight{3.6cm}
\setlength\figurewidth{7.0cm}
\center
\else
\setlength\figureheight{2.3 cm}
\setlength\figurewidth{3cm}
\center
\fi
\input{media/Plots3/abspos3.tikz}
\end{minipage}\hfill
\begin{minipage}[r]{.54\columnwidth}
\ifsingleColumn
\setlength\figureheight{3.6 cm}
\setlength\figurewidth{7.0cm}
\center
\else
\setlength\figureheight{2.3 cm}
\setlength\figurewidth{3.3cm}
\center
\fi
\input{media/Plots20/abspos20.tikz}
\end{minipage}
\hspace*{-8pt}
\caption{\HChange{Left: Platoon with three vehicles, where the evolution of the absolute positions of vehicle 1 (yellow), vehicle 2 (red), and vehicle 3 (blue) is shown. Right: Platoon with 20 vehicles. The vehicles are initialized with an inter-vehicle distance of $\unit[20]{m}$.}}
\label{Fig:Traj3}
\end{figure}



\begin{figure}
\begin{minipage}[l]{.5\columnwidth}
\ifsingleColumn
\setlength\figureheight{3.6 cm}
\setlength\figurewidth{7.0cm}
\center
\else
\setlength\figureheight{2.3 cm}
\setlength\figurewidth{3.3cm}
\center
\fi
\input{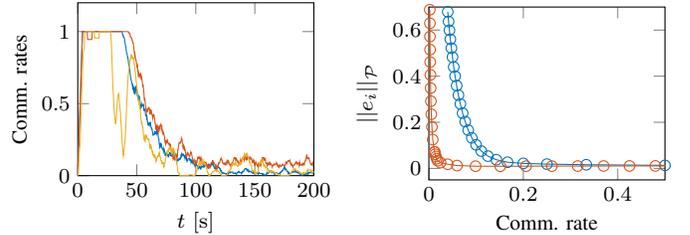}
\end{minipage}
\begin{minipage}[r]{.5\columnwidth}
\ifsingleColumn
\setlength\figureheight{3.6 cm}
\setlength\figurewidth{7.0cm}
\center
\else
\setlength\figureheight{2.3 cm}
\setlength\figurewidth{3.3cm}
\center
\fi
%
%
\definecolor{mycolor1}{rgb}{0.00000,0.44700,0.74100}%
\definecolor{mycolor2}{rgb}{0.85000,0.32500,0.09800}%
\begin{tikzpicture}

\begin{axis}[%
width=0.95092\figurewidth,
height=\figureheight,
at={(0\figurewidth,0\figureheight)},
scale only axis,
xmin=0,
xmax=.5,
xlabel={Comm. rate},
xlabel near ticks,
ymin=-0.05,
ymax=0.7,
ylabel={$||e_i||_\mathcal{P}$},
ylabel near ticks
]
\addplot [color=mycolor1,solid,mark=o,mark options={solid},forget plot]
  table[row sep=crcr]{%
1	0.00988749070813044\\
0.499980000799968	0.0130096057428977\\
0.333346666133355	0.0157018686104381\\
0.249990000399984	0.0184853730195692\\
0.199992000319987	0.0215266540486304\\
0.166673333066677	0.0266594124400599\\
0.142874285028599	0.0375003923084142\\
0.124995000199992	0.0531939494838069\\
0.111115555377785	0.0749769229230476\\
0.0999960001599936	0.102223222269453\\
0.0909163633454662	0.132495972687437\\
0.0833166673333067	0.165797897156844\\
0.0769169233230671	0.202931260601052\\
0.0714371425142994	0.239518563642342\\
0.0666773329066837	0.279773313717899\\
0.062517499300028	0.317951793309529\\
0.0588376464941402	0.358995254334838\\
0.0555577776888925	0.399684649744078\\
0.0526378944842206	0.439812687977266\\
0.0499980000799968	0.482964851205435\\
0.0476380944762209	0.521771529808857\\
0.0454381824727011	0.560090030803738\\
0.0434782608695652	0.599284865630065\\
0.0416783328666853	0.637058507392177\\
0.0399984000639975	0.679333812113199\\
};
\addplot [color=mycolor2,solid,mark=o,mark options={solid},forget plot]
  table[row sep=crcr]{%
0.99998733383998	0.0110050217309146\\
0.973195072197112	0.0110063141718453\\
0.918573923709718	0.0110106314476231\\
0.863938109142301	0.0110129331736696\\
0.809447622095116	0.0110087833954988\\
0.754451155287122	0.0109861512980537\\
0.699491353679186	0.0109549696840082\\
0.64436222551098	0.0109158978384561\\
0.589740410383585	0.0108405451137218\\
0.53489393757583	0.0107452000219472\\
0.47973747716758	0.0106673890894912\\
0.424941002359906	0.010508763302995\\
0.370111862192179	0.0102640163771328\\
0.315202725224324	0.0100667834808358\\
0.260455581776729	0.00979991181747134\\
0.205795101529272	0.00946160521585086\\
0.151259949602016	0.00920941174839582\\
0.0988500459981601	0.00940164802553727\\
0.0596336146554138	0.0121077825418049\\
0.0398457395037532	0.017828572280063\\
0.0246096822793755	0.0249546723272445\\
0.0186352545898164	0.0331594857629533\\
0.0165253389864405	0.0420024976050121\\
0.0149947335439916	0.0514740738777693\\
0.0132541365012066	0.0611690463179635\\
0.0121095156193752	0.0715223180537594\\
0.0121095156193752	0.0715223180537594\\
0.00814434089303095	0.124980362191441\\
0.00647640761036225	0.17985757981222\\
0.00509379624815007	0.23488706841721\\
0.00425982960681573	0.291429694458996\\
0.00369051904590483	0.348110126306379\\
0.00322387104515819	0.405710848826418\\
0.00282122048451395	0.46146869486902\\
0.00251589936402544	0.515047128354278\\
0.00222324440355719	0.571982314913065\\
0.0021552471234484	0.629424153855318\\
0.00189259096302815	0.688472782341245\\
};
\end{axis}
\end{tikzpicture}%
\end{minipage}\vspace*{-6pt}
\caption{\HChange{Left: Communication rates corresponding to the three vehicles in Fig.~\ref{Fig:Traj3}. Right: Performance versus communication plot for the event-based design (red) and the centralized design with reduced sampling rates (blue). The graph focuses on communication rates below 0.5, as the achieved performance $||e_i||_\mathcal{P}$ changes only insignificantly for rates above 0.5.}}
\label{Fig:PerfCar}
\end{figure}

\subsubsection{20 Vehicles}
The design procedure is repeated for the case $M=20$, which results in an optimization including $1973$ variables. For this example, the inter-agent error stability conditions provided by Cor.~\ref{Cor:IES2} would lead to a numerically intractable problem (this would amount to $2^{20}$ LMIs). 

The resulting closed-loop performance is evaluated in simulations, where the leading car is initialized with a surplus velocity of $\unitfrac[5]{m}{s}$, the state estimates of the different vehicles are initialized with zero, and again a packet loss rate of $10\%$ is introduced. The absolute positions of all vehicles 
are shown in Fig.~\ref{Fig:Traj3}. \HChange{In steady state, the distance error remains below $\unit[0.2]{m}$ for all 20 vehicles.} The communication rates are found to be higher for the leading vehicles, see Fig.~\ref{Fig:SSComm}, which can be explained by the fact that the actions of the leading vehicles influence all remaining vehicles.


\begin{figure}
\ifsingleColumn
\setlength\figureheight{4cm}
\setlength\figurewidth{7.7cm}
\center
\else
\setlength\figureheight{2.1cm}
\setlength\figurewidth{6.7cm}
\center
\vspace*{-8pt}
\fi
%
%
\begin{tikzpicture}

\begin{axis}[%
width=0.95092\figurewidth,
height=\figureheight,
at={(0\figurewidth,0\figureheight)},
scale only axis,
log origin=infty,
area legend,
xmin=0,
xmax=21,
xmajorgrids,
xlabel={Vehicle nr.},
xlabel near ticks,
ymode=log,
ymin=0.001,
ymax=1,
yminorticks=true,
ymajorgrids,
yminorgrids,
ylabel={Comm. rate},
ylabel near ticks
]
\addplot[ybar,bar width=0.030429\figurewidth,draw=black,fill=black] plot table[row sep=crcr] {%
1	0.98540129197416\\
2	0.27967440651187\\
3	0.197227055458891\\
4	0.0963390732185356\\
5	0.0328743425131497\\
6	0.023137537249255\\
7	0.0188616227675446\\
8	0.0166856662866743\\
9	0.0145517089658207\\
10	0.0134537309253815\\
11	0.0121337573248535\\
12	0.0129747405051899\\
13	0.0131107377852443\\
14	0.0121687566248675\\
15	0.0119557608847823\\
16	0.0125457490850183\\
17	0.0122587548249035\\
18	0.0112467750644987\\
19	0.0105477890442191\\
20	0.00978480430391392\\
};
\addplot [color=red,only marks,mark=.,mark options={solid},forget plot]
 plot [error bars/.cd, y dir = both, y explicit]
 table[row sep=crcr, y error plus index=2, y error minus index=3]{%
1	0.98540129197416	0.000406801013594798	0.000406801013594798\\
2	0.27967440651187	0.00268331420970574	0.00268331420970574\\
3	0.197227055458891	0.00151495415454567	0.00151495415454567\\
4	0.0963390732185356	0.00114926466593381	0.00114926466593381\\
5	0.0328743425131497	0.0022206133850148	0.0022206133850148\\
6	0.023137537249255	0.0021475520318475	0.0021475520318475\\
7	0.0188616227675446	0.00181332815525502	0.00181332815525502\\
8	0.0166856662866743	0.00197331034406098	0.00197331034406098\\
9	0.0145517089658207	0.00175202119515519	0.00175202119515519\\
10	0.0134537309253815	0.00171400803338692	0.00171400803338692\\
11	0.0121337573248535	0.00147255653063215	0.00147255653063215\\
12	0.0129747405051899	0.00236724624636628	0.00236724624636628\\
13	0.0131107377852443	0.00165542008125911	0.00165542008125911\\
14	0.0121687566248675	0.00157272977799407	0.00157272977799407\\
15	0.0119557608847823	0.00226985467026407	0.00226985467026407\\
16	0.0125457490850183	0.00180217429531008	0.00180217429531008\\
17	0.0122587548249035	0.00175533529954959	0.00175533529954959\\
18	0.0112467750644987	0.00220405043411061	0.00220405043411061\\
19	0.0105477890442191	0.00184883636441292	0.00184883636441292\\
20	0.00978480430391392	0.00135969220445295	0.00135969220445295\\
};
\end{axis}
\end{tikzpicture}%
\vspace*{-11pt}
\caption{\HChange{Communication at steady state for all 20 vehicles. The communication rates are determined by $20$ independent simulations (different noise realizations) of the system over a time horizon of $\unit[1000]{s}$, which was found to be sufficiently long for transients to die out.} The error bars indicate the standard deviation over the different noise realizations.}
\label{Fig:SSComm}
\end{figure}
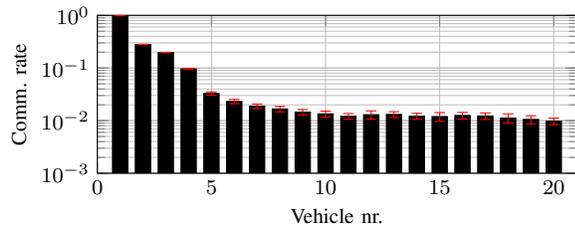
\vspace*{-6pt}\section{Conclusion}\label{Sec:Conclusion}
A synthesis procedure for the estimator gains and the communication thresholds for a distributed event-based control system was presented in this article. In contrast to previous work, communication is reduced by avoiding the exchange of control inputs between agents and (under favorable circumstances) estimator resets. The design guarantees a stable closed-loop system and satisfies a predefined worst-case estimation performance. Relaxations leading to a simplified design
 are particularly suitable for a large number of agents.
The proposed event-based design was evaluated in simulations of a vehicle platoon.
\ifsingleColumn
\else
\vspace{-11pt}
\fi


%

%

%
%

\ifCLASSOPTIONcaptionsoff
  \newpage
\fi



\bibliographystyle{IEEEtran}
\bibliography{literature}

\ifArxiv
\newpage

\appendices

\section{Proof of Lemma~\ref{Lem:ES}}\label{App:LeES}
\begin{proof} \emph{Sufficiency:} Let the matrix $(I-LC)A$ have eigenvalues with magnitude strictly less than one. According to \eqref{eq:AgentErrorDyn} it is enough to show that $\xi$ is bounded, since $\epsilon_{ji}$, $w_i$, $\xi$, and $d_i$ are bounded by assumption. From the triangle inequality, the submultiplicativity of the two-norm, and the communication protocol \eqref{eq:I}, it follows that $|\xi(k)|$ is bounded by
\begin{align*}
\sum_{i \in \comp{I}(k)}\!\!\!\!|L_i \Delta_i| |\Delta_i^{-1} (y_i(k)-C_i \hat{x}_i(k|k\!-\!1))| \! \leq \!\sum_{i=1}^N |L_i \Delta_i|.
\end{align*} 
\emph{Necessity:} The argument is based on contradiction. Thus we assume the system to be ISS and the matrix $(I-LC)A$ to have at least one eigenvalue of magnitude greater or equal than one. Choosing disturbances $d_i$ parallel to an eigenvector of $(I-LC)A$ with corresponding eigenvalue having magnitude greater or equal than one contradicts the assumption that the agent error is ISS.
\end{proof}

\section{Proof of Thm.~\ref{Thm:Performance}}\label{App:ThmVI1}
\begin{proof} 
The dynamics of the performance objective $z$, as defined in \eqref{eq:DefZ}, can be written as
\begin{align}
\begin{split}
e_i(k)&=\hat{A} e_i(k-1) + L \Delta s_1(k) + \hat{B}_{2} s_{2}(k) \\
z(k)&=\hat{C} e_i(k-1) + \hat{D}_{2} s_{2}(k),
\end{split}
\end{align}
where
\begin{align*}
s_{2}(k)&:=\begin{bmatrix} W^{-\frac{1}{2}} w(k)\\ V^{-\frac{1}{2}} v(k-1)\end{bmatrix}.
\end{align*}
The communication protocol guarantees that $|s_{1i}(k)|$ is strictly less than one and therefore $|s_1(k)| < \sqrt{N}$.

Let the impulse response from $s_1$ to $z$ be denoted by $g_1$ and the impulse response from $s_{2}$ to $z$ by $g_{2}$. Both are well defined, since the matrix $\hat{A}$ has eigenvalues strictly within the unit circle, which is implied by the first matrix inequality in \eqref{eq:H2lmi}. Using the fact that $||\cdot||_\mathcal{P}$ is a semi-norm yields
\begin{equation}
||z||_\mathcal{P} \leq ||g_1 * s_1 ||_\mathcal{P} + ||g_{2} * s_{2}||_\mathcal{P}, \label{eq:BoundZ}
\end{equation}
where $*$ denotes the convolution operator. The first term can be upper bounded by, \cite[p. 107]{zhou1996robust}\footnote{A continuous-time derivation is presented in \cite{zhou1996robust}. The discrete-time case used herein is analogous.}
\begin{equation}
||g_1 * s_1 ||_\mathcal{P} \leq ||G_1 ||_\infty ||s_1||_\mathcal{P} \leq ||G_1||_\infty \sqrt{N}, \label{eq:G1}
\end{equation}
whereas the second term yields
$||g_{2} * s_{2}||_\mathcal{P}=||G_{2}||_2$, 
by the statistical properties of $s_{2}$, \cite[p. 108]{zhou1996robust}. Note that $G_1$ and $G_{2}$ represent the Z-transforms of $g_1$, respectively $g_{2}$, $||G_1||_\infty$ the $\mathcal{H}_\infty$ norm of $G_1$, and $||G_{2}||_2$ the $\mathcal{H}_2$ norm of $G_{2}$, see e.g. \cite[pp. 97-100]{zhou1996robust}. Thus, combining \eqref{eq:BoundZ} and \eqref{eq:G1} yields
\begin{equation}
||z||_\mathcal{P} \leq ||G_1||_\infty \sqrt{N} + ||G_{2}||_2. \label{eq:prooftmp}
\end{equation}
According to \cite[Lemma 2]{de2002extended}, it holds that $||G_1||_\infty < \sqrt{\gamma}$, where $\gamma \in \mathbb{R}$ satisfies \eqref{eq:Hinflmi}, and according to \cite[Theorem A.2 (Appendix)]{muehlebach2015guaranteed}, $||G_{2}||_2 < \sqrt{\tr{X}}$ holds, where $X=X\T$ satisfies \eqref{eq:H2lmi}.\hfill
\end{proof}
\section{Continuous Local Measurement Update}\label{Subsec:LocMeasUp}
According to \eqref{eq:eventTrigger}, \eqref{eq:I}, 
the measurement $y_i(k)$ is used in the estimator update \eqref{eq:EBSE2} only
if the condition $|\Delta_i^{-1} (y_i(k)-C_i \hat{x}_i(k|k-1))| \geq 1$ is satisfied. However, each agent could  include its local measurements $y_i$ in the update \eqref{eq:EBSE2} continuously (irrespective of the event trigger) without requiring additional communication.
The implications of this alternative scheme regarding closed-loop stability are analyzed next. 

For each agent $i$, let the indicator function $\chi_{i\in \comp{I}}(k)$ be defined as $\chi_{i\in \comp{I}}(k)=1$ if $i\in \comp{I}(k)$ and $0$ otherwise, for $k\in \mathbb{N}$.
In case each agent continuously updates its state estimate with local measurements, the estimation update \eqref{eq:EBSE2} is replaced by
\begin{align}
\hat{x}_i(k)&=\hat{x}_i(k|k-1)+ \sum_{j \in I(k)} L_j (y_j(k)-C_j \hat{x}_i(k|k-1)) \nonumber\\
&+ \underbrace{\chi_{i\in \comp{I}}(k) L_i (y_i(k)-C_i \hat{x}_i(k|k-1)) + d_i(k)}_{:=\bar{d}_i(k)},
\end{align}
where the additional term can be regarded as a disturbance and forms, together with $d_i(k)$, the disturbance $\bar{d}_i(k)$. In fact, $|\bar{d}_i|$ is bounded by
\ifsingleColumn
\begin{equation}
|d_i(k)|+ \chi_{i\in \comp{I}}(k) |L_i| |y_i(k)-C_i \hat{x}_i(k|k-1)| 
< |d_i(k)| + |L_i| \sigma_\text{max} (\Delta_i),
\end{equation}
\else
\begin{align}
\begin{split}
|d_i(k)|+ \chi_{i\in \comp{I}}(k) |L_i|& |y_i(k)-C_i \hat{x}_i(k|k-1)| \\
< &|d_i(k)| + |L_i| \sigma_\text{max} (\Delta_i),
\end{split}
\end{align}
\fi
since $\chi_{i\in \comp{I}}(k)=1$ implies
\begin{equation}
|y_i(k)-C_i \hat{x}_i(k|k-1)| < \sigma_\text{max}(\Delta_i). 
\end{equation}
Hence, the conditions ensuring ISS of the closed-loop system established previously remain valid even in case each agent continuously updates his state estimate with local measurements. While causing no additional communication, such a scheme potentially improves the estimation performance since each agent exploits all locally available measurements.

\section{Modeling Packet Drops}\label{App:ModPacketDrops}
If the communication from agent $m$ to agent $i$ fails at time $k$, the disturbance $d_i(k)$ takes the value
\begin{equation}
d_i(k)= - L_m (y_m(k)-C_m \hat{x}_i(k|k-1)). \label{eq:di1}
\end{equation}
As shown below, \eqref{eq:di1} is a function of the agent errors $e_i$, the process noise $v$, and the measurement noise $w_m$. Hence, if $d_i$ is used to model packet drops, it is implicitly dependent on the agent error $e_i$, and boundedness of $d_i$ cannot be guaranteed a priori. However, we will argue that the $d_i$'s are indeed bounded if packet drops are sufficiently rare, and the conditions given by Thm.~\ref{Thm:IES1}, Cor.~\ref{Cor:IES2}, or Cor.~\ref{Cor:IES3} are fulfilled. We provide a qualitative argument, which can be turned into a quantitative statement about the allowed frequency of packet drops so as to still
 guarantee boundedness of the disturbances $d_i$. Although these statements tend to be conservative, the simulation examples presented in Sec.~\ref{Sec:SimExamples} indicate that relatively frequent packet drops can be tolerated (e.g.\ packet loss probability of $10\%$).

We assume $d_i(1)$ arbitrary and $d_i(k)=0$ for all agents $i$ and for all $2\leq k \leq k_0$, where $k_0$ is a positive integer, describing the earliest time instant at which the next packet drop can occur. We therefore model the packet drops as being sufficiently rare, that is, the number of time instants between two consecutive packet drops is greater or equal than $k_0$.
We assume further that the conditions of Thm.~\ref{Thm:IES1} are fulfilled (the argument is analogous in case the conditions of Cor.~\ref{Cor:IES2} or Cor.~\ref{Cor:IES3} are satisfied). From \eqref{eq:ejidecay} it follows that the inter-agent error decays exponentially due to the fact that $d_i(k)=0$ for all $2\leq k \leq k_0$.
The agent-error can be regarded as a linear time-invariant system with system matrix $(I-LC)A$, which is Schur stable. Thus, an exponentially decaying input will lead to an exponentially decaying output. As a consequence, the agent-error $|e_i(k)|$ can be bounded by
\begin{equation}
a_1^k b_1 \sum_{j=1}^N |d_j(1)| + b_2,
\end{equation}
where $a_1 < 1$ is the decay rate and  $b_2$ is a constant depending on the bounds for $\xi$, $v$, $w$, and $|e_i(0)|$.

Provided that the communication from agent $m$ to agent $i$ fails at time $k$, the measurement equation in \eqref{eq:LinSys} can be used to rewrite \eqref{eq:di1} as
\begin{equation}
d_i(k)=-L_m C_m (x(k) - \hat{x}_i(k|k-1)) - L_m w_m(k),
\end{equation}
which leads, according to \eqref{eq:EBSE1}, \eqref{eq:EstFB}, and \eqref{eq:CLx}, to
\ifsingleColumn
\begin{equation}
d_i(k)=-L_m C_m [(A+BF) e_i(k-1) - \sum_{j=1}^{N} B_j F_j e_j(k-1) + v(k-1) ] - L_m w_m(k).
\end{equation}
\else
\begin{align}
\begin{split}
d_i(k)&=-L_m C_m [(A+BF) e_i(k-1) \\
&- \sum_{j=1}^{N} B_j F_j e_j(k-1) + v(k-1) ] - L_m w_m(k).
\end{split}
\end{align}
\fi
Given that packet drops happen at times $m k_0+1$, $m\in \mathbb{N}$ (or less frequent), we bound $|d_i(mk_0+1)|$ for all agents $i$ using a worst case upper bound over all possible communication failures; that is,
\begin{equation}
|d_i(m k_0+1)| \leq a_1^{k_0} b_3 \sum_{j=1}^{N} |d_j((m-1) k_0+1)| + b_4,
\end{equation}
where $b_3 >0$ and $b_4>0$ are constants. For large enough $k_0$, it follows that $a_1^{k_0} b_3 < 1/N$ and therefore
\begin{equation}
\sum_{i=1}^N |d_i(m k_0+1)| < \sum_{i=1}^N |d_i((m-1) k_0 +1)| + N b_4,
\end{equation}
for all $m\in \mathbb{N}$. Thus, if packet drops are sufficiently rare, the assumption that the disturbances $d_i$ are bounded is indeed valid.

\section{Feasibility}\label{App:Feas}
The stability conditions given by Thm.~\ref{Thm:IES1}, Cor.~\ref{Cor:IES2}, and Cor.~\ref{Cor:IES3} might be too restrictive, resulting in an infeasible synthesis problem (in Step 1). In this case, the inter-agent error is not guaranteed to be ISS. In \cite{Sebastian-CDC}, a reset strategy was introduced to periodically reset the inter-agent error using additional communication. In the following, an extension to this approach is provided ensuring input-to-state stability of the inter-agent error, even in case the corresponding LMI conditions are infeasible. We will use the conditions in Thm.~\ref{Thm:IES1} as starting point. The procedure is analogous if the conditions provided by Cor.~\ref{Cor:IES2}, and Cor.~\ref{Cor:IES3} are used to guarantee inter-agent error stability. 

In a first step, the conditions given by \eqref{eq:IESCond1} are relaxed to
\begin{equation}
A_\text{cl}\T(\Pi_i) P_k A_\text{cl}(\Pi_i) - P_l < \bar{\lambda} I, \label{eq:IESCondRelaxed}
\end{equation}
for all $\Pi_i \in \Pi$ with $k=1$ if $\emptyset \not\in \Pi_i$, $k=2$ if $\Pi_i=\{\emptyset \}$ and $l=1,2$. Note that $\bar{\lambda}\geq 0$ is either fixed, or can be included in the optimization problem as decision variable, see \cite{muehlebach2015lmi}.

From the proof of Thm.~\ref{Thm:CL}, it follows that the function $V$ in \eqref{eq:LyapFun} can be bounded by (c.f. \eqref{eq:evalV})
\begin{equation}
V(k)\leq \left(\frac{\bar{\lambda} + \alpha}{\munderbar\sigma}+1\right) V(k-1) + \left( \frac{\bar{\gamma}^2}{\alpha} + \bar{\delta} \right) D^2, \label{eq:FeasTemp}
\end{equation}
where $D$ is an upper bound to the disturbances $d_{ji}(k)$, i.e. $|d_{ji}(k)| \leq D$ for all $k$, $\bar \delta:= \max_{m \in \{1,2\}} |P_m|$, and $\bar \gamma:= \max_{\Pi_i \in \Pi, m\in\{1,2\}} |P_m A_{\text{cl}}(\Pi_i)|$. 
Therefore, an estimate $\hat{V}(k)$ with $\hat{V}(k) \geq V(k)$ is given by
\begin{align}
\hat{V}(k) = \left(\frac{\bar{\lambda} + \alpha}{\munderbar\sigma}+1\right) \hat{V}(k-1) + \left( \frac{\bar{\gamma}^2}{\alpha} + \bar{\delta} \right) D^2,
\end{align}
for $k \in \mathbb{N}$, $\hat{V}(0)=0$ (provided that all agents are initialized with the same state estimate). Note that in order to tighten the bound, the right hand side of \eqref{eq:FeasTemp} can be minimized with respect to $\alpha >0$, as done in \cite{muehlebach2015lmi}.

As soon as $\hat{V}$ exceeds the predefined threshold $V_\text{max}$, i.e. $\hat{V}(k)\geq V_\text{max}$, a communication is triggered and the different agents' state estimates are set to a common value, which resets the inter-agent errors $\epsilon_{ji}(k)=0$ and implies $\hat{V}(k)=0$.
There are many different reset strategies that can be used, such as a majority vote, the mean, etc.
Such resets bound the inter-agent error since $V_\text{max} \geq V(k) \geq \munderbar\sigma |\epsilon_{ji}(k)|^2$ for all $k$. By the strict feedforward structure of the closed-loop dynamics, this implies ISS of the state $x$ and the agents' estimation errors $e_i$, $i=1,2,\dots,N$. 
 The time instants $k_{\text{reset}_i}$, where $\hat{V}(k_{\text{reset}_i})$ exceeds $V_\text{max}$ for $i=1,2,\dots,N$
 can be precalculated, since the evolution of $\hat{V}(k)$ is not explicitly dependent on time. This amounts to periodic resets and extends the procedure presented in \cite{Sebastian-CDC} by providing a method for choosing the reset period.

The synthesis of the communication thresholds $\Delta_i$ in Step 2 is guaranteed to be feasible. This is because the full communication scenario can be recovered by making the thresholds $\Delta_i$ arbitrarily small; that is, 
$\gamma \rightarrow 0$ in Thm.~\ref{Thm:Performance} (we refer to \cite{muehlebach2015guaranteed} for further details).

\section{Communication of the Inputs}\label{App:ComIn}
In case of an unstable open-loop system, it is essential for guaranteeing inter-agent stability that each agent reconstructs the input $u$ based on its current state estimate $\hat{x}_i$, as opposed to the case where all agents have access to the true input $u$ (proposed in \cite{Sebastian-CDC,Tr12}).

The mechanism leading to a destabilization in case the inputs are communicated can be illustrated by a simple two-agent system having an unstable mode, which is only controllable by agent 1, and only observable by agent 2. Roughly speaking, in case the agents cannot access the true inputs, the inter-agent error tends to decay (e.g. in case there is no communication according to the stable closed-loop dynamics $A+BF$) resulting in communication by agent 2 if the predicted and actual measurements are too far apart, thereby stabilizing the system. In case the agents have access to the true inputs, agent 2 might observe the unstable mode perfectly (but cannot control it), and might thus never share a measurement with agent 1 (who cannot observe the unstable mode at all, but would be able to control it). 

Specifically, this mechanism can be illustrated on the system with matrices
\begin{align}
A&=\left( \begin{array}{cc} 0.5 & 0 \\ 0 & 2\end{array}\right), ~B=\left(\begin{array}{cc} 0 & 1 \\ 1 & 0 \end{array} \right), ~C=\left( \begin{array}{cc} 1 & 0\\ 0& 1\end{array} \right), \\
F&=\left(\begin{array}{cc} 0 & -2 \\ 0.1 & 0\end{array}\right), ~L=\left( \begin{array}{cc} 1 & 0 \\ 0 & 1 \end{array} \right),
\end{align}
where the first agent measures the first component of $y$ and controls the first component of $u$, and the second agent measures the second component of $y$ and controls the second component of $u$. Clearly, the matrices
\begin{align}
A+BF=\left( \begin{array}{cc} 0.6 & 0 \\ 0 & 0 \end{array} \right), \quad (I-LC)A=0,
\end{align}
are stable. The initial condition are chosen as
\begin{align}
\hat{x}_1(0)=\left( \begin{array}{c} 0 \\ 1\end{array} \right), \quad \hat{x}_2(0)=\left( \begin{array}{c} 0 \\ 2\end{array} \right), \quad x(0)=\left( \begin{array}{c} 0 \\ 2\end{array} \right),
\end{align}
and for simplicity, it is assumed that there is neither process noise nor measurement noise, and that the communication thresholds $\Delta_1$ and $\Delta_2$ are set to 1.

In case the input $u$ is communicated, the following sequences of inputs, states, and estimates is obtained
\begin{align*}
\text{Step 1:}\quad &u(0)=\VecT[-2]{0}, ~x(1)=\VecT[0]{2}, ~y(1)=\VecT[0]{2},\\
&\hat{x}_1(1|0)=\VecT[0]{0}, ~\hat{x}_2(1|0)=\VecT[0]{2} \\ &\xrightarrow{\text{no comm.}} \hat{x}_1(1)=\VecT[0]{0}, ~\hat{x}_2(1)=\VecT[0]{2} \\
\text{Step 2:}\quad &u(1)=\VecT[0]{0}, ~x(2)=\VecT[0]{4}, ~y(2)=\VecT[0]{4}, \\
&\hat{x}_1(2|1)=\VecT[0]{0}, ~\hat{x}_2(2|1)=\VecT[0]{4} \\
&\xrightarrow{\text{no comm.}} \hat{x}_1(2)=\VecT[0]{0}, ~\hat{x}_2(2)=\VecT[0]{4},
\end{align*}
leading to $u(n)=\hat{x}_1(n)=0$ and
\begin{equation}
x(n)=y(n)=\hat{x}_1(n)=\hat{x}_2(n)=\VecT[0]{2^n},
\end{equation}
for all $n>0$. Agent 2, which can observe the unstable mode $x_2$, tracks the state perfectly, and as a result, will never communicate its local measurements $y_2$. In contrast, agent 1, which could control the unstable mode, obtains no information about $x_2$. Thus, in the above example, the state estimate $\hat{x}_1$ will stay at zero for all times, whereas $\hat{x}_2$ tracks $x$ perfectly. Overall an unstable closed-loop system is obtained, unless a periodic estimator reset (as proposed in \cite{Sebastian-CDC}) is introduced. Such a reset strategy will periodically set the agents' state estimates to a common average, thereby providing agent 1 with information about $x_2$, resulting in a stabilization of the closed-loop system, as shown in \cite{Sebastian-CDC}.

In case the input is not communicated, the following evolution of the closed-loop system is obtained
\begin{align*}
\text{Step 1:}\quad&u(0)=\VecT[-2]{0}, ~x(1)=\VecT[0]{2}, ~y(1)=\VecT[0]{2}\\
&\hat{x}_1(1|0)=\VecT[0]{0}, ~\hat{x}_2(1|0)=\VecT[0]{0} \\
&\xrightarrow{\text{agent 2 comm}} \hat{x}_1(1)=\VecT[0]{2}, ~\hat{x}_2(1)=\VecT[0]{2}\\
\text{Step 2:}\quad&u(1)=\VecT[0]{-4}, ~x(2)=\VecT[0]{0}, ~y(2)=\VecT[0]{0}\\
&\hat{x}_1(2|1)=\VecT[0]{0}, ~\hat{x}_2(2|1)=\VecT[0]{0} \\
&\xrightarrow{\text{no comm.}} \hat{x}_1(2)=\VecT[0]{0}, ~\hat{x}_2(2)=\VecT[0]{0},
\end{align*}
leading to $u(n)=x(n)=\hat{x}_1(n)=\hat{x}_2(n)=0$ for all $n>1$. In that case, both agents track the state perfectly, because agent 2 communicates its measurement $y_2(1)$ and thus shares its information about the unstable mode with agent 1 who is able to drive the system to 0. Thus, by not sharing the inputs, a stable closed-loop system is obtained. The conditions from Cor.~V.2 are clearly fulfilled, as the Lyapunov matrix $P$ can, for example, be chosen to be the identity. Thus, according to Thm.~V.5 the closed-loop system is guaranteed to be stable.

\begin{figure}
\begin{minipage}[l]{.2\columnwidth}
\graphicspath{{media/}}
\centering
\ifsingleColumn
    \def\svgwidth{0.7\columnwidth} 
\else
    \def\svgwidth{0.9\columnwidth} 
\fi
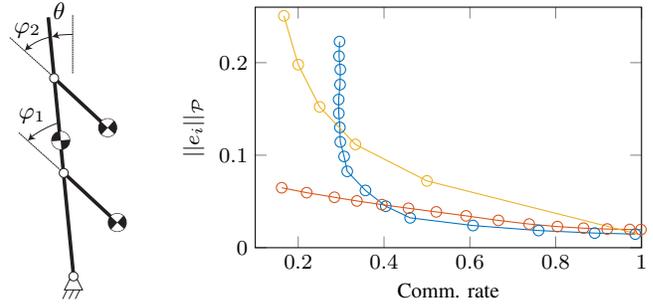
\end{minipage}\hspace{.5cm}
\begin{minipage}[r]{.8\columnwidth}
\ifsingleColumn
\setlength\figureheight{5.2cm}
\setlength\figurewidth{7.4cm}
\center
\else
\setlength\figureheight{3.2cm}
\setlength\figurewidth{5.4cm}
\fi
%
%
\definecolor{mycolor1}{rgb}{0.00000,0.44700,0.74100}%
\definecolor{mycolor2}{rgb}{0.85000,0.32500,0.09800}%
\definecolor{mycolor3}{rgb}{0.92900,0.69400,0.12500}%
\definecolor{mycolor4}{rgb}{1,0.6,0.6}%
\begin{tikzpicture}

\begin{axis}[%
width=0.95092\figurewidth,
height=\figureheight,
at={(0\figurewidth,0\figureheight)},
scale only axis,
xmin=0.1,
xmax=1.0,
xlabel={Comm. rate},
xlabel near ticks,
ymin=0,
ymax=0.26,
ylabel={$||e_i||_{\mathcal{P}}$},
ylabel near ticks
]
\addplot [color=mycolor1,solid,mark=o,mark options={solid},forget plot]
  table[row sep=crcr]{%
0.98552	0.0147020303593256\\
0.890838333333334	0.0158224439935142\\
0.760108333333333	0.0185523933408206\\
0.607273333333333	0.0239553979784145\\
0.46084	0.0322738308239057\\
0.404076666666667	0.0448392699139189\\
0.356763333333333	0.0618948636666186\\
0.314248333333333	0.082521911939148\\
0.306985	0.0986151246538441\\
0.298103333333333	0.114318483227784\\
0.297078333333333	0.130019442571802\\
0.29462	0.145210940833157\\
0.294245	0.160166301476896\\
0.297553333333333	0.176325904085348\\
0.298121666666667	0.192596236635813\\
0.29512	0.207006449257957\\
0.296341666666667	0.222826726950939\\
};
\addplot [color=mycolor2,solid,mark=o,mark options={solid},forget plot]
  table[row sep=crcr]{%
0.998518333333334	0.0195536576644157\\
0.97263	0.0197069129700827\\
0.919941666666667	0.0201834270159706\\
0.864843333333334	0.0211039046289525\\
0.804026666666667	0.0228178000899266\\
0.738181666666667	0.0255593244425293\\
0.666468333333333	0.0295236244287998\\
0.591216666666667	0.0344316812848818\\
0.521756666666667	0.0387059876185621\\
0.456951666666667	0.0426317415013527\\
0.394936666666667	0.0465493656511884\\
0.336755	0.0505526695028366\\
0.284351666666667	0.0545114040022735\\
0.21967	0.0594907349705514\\
0.161515	0.0648204121456059\\
};
\addplot [color=mycolor3,solid,mark=o,mark options={solid},forget plot]
  table[row sep=crcr]{%
1.00006666666667	0.0138261968609451\\
0.5	0.072214987910158\\
0.333333333333333	0.111633797296033\\
0.25	0.152219438404884\\
0.2	0.197920003792878\\
0.166666666666667	0.250677616272214\\
};
\end{axis}
\end{tikzpicture}%
\end{minipage}
\caption{Inverted pendulum balanced by two independently controlled arms (left), 
and resulting performance versus communication plots for different event-based estimator designs (right). Blue: event-based design with the stability conditions of Cor.~\ref{Cor:IES2}; Red: event-based design with the less conservative conditions of Thm.~\ref{Thm:IES1}; Yellow: centralized design with reduced sampling rates. }
\label{Fig:PerfPend}
\end{figure}


\section{Inverted Pendulum System}\label{App:Sim}
The example is taken from \cite{Sebastian-CDC}, where it was proposed as an abstraction of the Balancing Cube \cite{trimpe2012balancing}, which was the experimental test bed for the distributed and event-based methods in \cite{Tr12} and \cite{TrDAn11}.
\HChange{The pendulum system is parametrized by the inclination angle $\theta$, the angle $\varphi_1$ of the lower arm (called Agent 1), and the angle $\varphi_2$ of the upper arm (Agent 2), see Fig.~\ref{Fig:PerfPend}.} 
A state-space model \eqref{eq:LinSys} is obtained through discretization of the continuous dynamics with a sampling time of \unit[10]{ms}. The state is given by $x^{\mathrm{T}}=(\theta,\dot{\theta},\varphi_1,\dot{\varphi}_1,\varphi_2,\dot{\varphi}_2)$, and the inputs are the desired angular rates for the arms, $u=(\dot{\varphi}_{1\text{des}}, \dot{\varphi}_{2\text{des}})$. We refer to \cite{muehlebach2015lmi} for details of the modeling and the numerical values of the state-space matrices.

\HChange{Agent 1 measures $\varphi_1 + w_{\varphi_1}$, $\dot{\varphi}_1+w_{\dot{\varphi}_1}$, and $\dot{\theta}+w_{\dot{\theta}}$; and controls $u_1=\dot{\varphi}_{1\text{des}}+v_{u_1}$. Agent 2 measures $\varphi_2 + n_{\varphi_2}$ and $\dot{\varphi}_2+v_{\dot{\varphi}_2}$; and controls $u_2=\dot{\varphi}_{2\text{des}}+v_{u_2}$.} The signals $v_{\varphi_1}$, $v_{\varphi_2}$, $v_{\dot{\varphi}_1}$, $v_{\dot{\varphi}_2}$, $v_{\dot{\theta}}$, $w_{u_1}$, and $w_{u_2}$ are assumed to be independent, uniformly distributed with zero mean and variances $\sigma_{\varphi_i}^2=(\unit[0.05]{^\circ})^2$, $\sigma_{\dot{\varphi}_i}^2=(\unitfrac[0.1]{^\circ}{s})^2$, $\sigma_{\dot{\theta}}^2=(\unitfrac[0.24]{^\circ}{s})^2$, $\sigma_{u_i}^2=(\unitfrac[1.73]{^\circ}{s})^2$, $i=1,2$. Note that both measurement noise and input noise are introduced. A packet loss probability of $10\%$ is assumed (independent Bernoulli-distributed). \HChange{The simulation results indicate that the approach is robust also to non-deterministic and potentially unbounded disturbances $d_i$.}

The system is controllable and observable, but neither controllable nor observable for each agent on its own. In order to stabilize the upright equilibrium, communication between the agents is indispensable.

A stabilizing state feedback controller $F$ is obtained via a linear quadratic regulator approach, whose values can be found in \cite{muehlebach2015lmi}.

As performance measure, the power of the agent-error $e_i$ is used. Observer gains and communication thresholds are synthesized according to Sec.~\ref{Sec:SubsecSyn}. The optimizations are solved up to a tolerance of $10^{-8}$ using SDPT-3, \cite{Toh99sdpt3}, interfaced through Yalmip, \cite{JohanYALMIP}.



For the disturbance rejection properties of an event-based design based on Cor.~\ref{Cor:IES2}, and a design primarily aimed at reducing communication, we refer to \cite{muehlebach2015lmi}, respectively \cite{muehlebach2015guaranteed}. Herein, we focus on the trade-off between estimation performance and communication, which is obtained by varying $J_\text{max}$. The steady-state performance $||e_i||_\mathcal{P}$ and the communication rates of the different designs (obtained by successively increasing $J_\text{max}$) are evaluated in simulations. Their values were estimated using $20$ independent simulations (with different noise realizations) over $\unit[150]{s}$. The variability among the different noise realizations was found to be negligible and a time horizon of $\unit[150]{s}$ sufficiently long for transients to be insignificant. The communication versus performance graphs, resulting from the different designs, i.e. stability conditions according to Thm.~\ref{Thm:IES1} and Cor.~\ref{Cor:IES2}, are depicted in Fig.~\ref{Fig:PerfPend} (right), which also includes the graph for a centralized discrete-time design with reduced sampling rates for comparison. \HChange{As in Sec.~\ref{Sec:SimExamples},  the centralized design is obtained by re-sampling the discrete-time system \eqref{eq:LinSys} at increasingly lower rates, and then performing a centralized steady-state Kalman filter design based on the performance objective $||e_i||_\mathcal{P}$. The fact that the inputs are also communicated is not accounted for in the corresponding communication rates shown in Fig.~\ref{Fig:PerfPend}.} The communication rate is normalized such that a rate of 1.0 corresponds to both agents transmitting their measurements at every time step.

The comparison in Fig.~\ref{Fig:PerfPend} reveals that for communication rates above $40\%$ the design based on the stability conditions given by Cor.~\ref{Cor:IES2} is superior. In case the communication is further reduced, but kept above $15\%$, the design based on the stability conditions given by Thm.~\ref{Thm:IES1} achieves a lower cost. \HChange{If $J_{\text{max}}$ is increased further, the communication rate is found to increase again, which is possibly due to nonlinear effects.}
\HChange{Compared to the centralized design with reduced periodic sampling rates a better trade-off is achieved by the event-based design.}



\else
\fi

\end{document}